\documentclass[copyright, creativecommons]{eptcs}
\providecommand{\event}{QPL 2023} 
\usepackage[utf8]{inputenc}
\usepackage{tikz-cd}
\usepackage{tikzit}
\usepackage{amsmath}
\usepackage{amsthm}
\usepackage{amssymb}
\usepackage{braket}
\usepackage{appendix}
\usepackage{hyperref}
\usepackage{stmaryrd}
\usepackage{float} 

\usepackage[utf8]{inputenc}
\usepackage[english]{babel}
\usepackage[T1]{fontenc}
\usepackage[numbers,sort&compress]{natbib}
\usepackage{stmaryrd} 
\usepackage{mathtools}

\usepackage[shortlabels]{enumitem}

\usepackage{xspace}
\usepackage{algorithm}
\usepackage{algorithmicx}
\usepackage[noend]{algpseudocode}
\usepackage{verbatim}

\newcommand{\symd}{\mathbin{\Delta}\xspace}
\newcommand{\Symdi}[1]{\underset{\scriptstyle #1}{\scalebox{1.5}{$\symd$}}\,}

\makeatletter
\newcommand\etc{etc\@ifnextchar.{}{.\@}\xspace}

\makeatother

\newcommand{\odd}[1]{\mathsf{Odd}\left(#1\right)}

\newcommand{\odds}[2]{\mathsf{Odd}_{#1}\left(#2\right)}

\newcommand{\intf}[1]{\left\llbracket #1 \right\rrbracket} 

\newcommand{\ld}{\lambda}
\newcommand{\sse}{\subseteq}

\newcommand{\XYm}{\ensuremath\normalfont\textrm{XY}\xspace}
\newcommand{\XZm}{\normalfont\normalfont\textrm{XZ}\xspace}
\newcommand{\YZm}{\normalfont\normalfont\textrm{YZ}\xspace}

\newcommand{\LOG}{labelled open graph}

\usepackage{bm}

\newcommand{\abs}[1]{\ensuremath{\left| #1 \right|}}
\usepackage{amsmath,amsthm,amssymb}

\usepackage[color,leftbars]{changebar}

\definecolor{zx_red}{RGB}{232, 165, 165}
\definecolor{zx_green}{RGB}{216, 248, 216}

\tikzstyle{gn}=[rectangle,rounded corners=0.8em,fill=zx_green,draw=black,
  line width=0.8 pt,inner sep=3pt,minimum width=1.5em,minimum height=1.5em]
\tikzstyle{rn}=[rectangle,rounded corners=0.8em,fill=zx_red,draw=black,
  line width=0.8 pt,inner sep=3pt,minimum width=1.5em,minimum height=1.5em]
\tikzstyle{medium box}=[fill=white, draw=black, shape=rectangle, minimum width=0.75cm, minimum height=1cm]
\tikzstyle{small box}=[fill=white, draw=black, shape=rectangle, minimum width=0.5cm, minimum height=0.66cm]
\tikzstyle{H gate}=[fill={rgb,255: red,255; green,242; blue,52}, draw=black, shape=rectangle, minimum width=0.25cm, minimum height=0.25cm]
\tikzstyle{blackdot}=[fill=black, draw=black, shape=circle]

\tikzstyle{Hadamard edge}=[-, dashed, dash pattern=on 2pt off 1.5pt, thick, draw={rgb,255: red,68; green,136; blue,255}]

\tikzstyle{box}=[shape=rectangle, text height=1.5ex, text depth=0.25ex, yshift=0.5mm, fill=white, draw=black, minimum height=5mm, yshift=-0.5mm, minimum width=5mm, font={\small}]
\tikzstyle{Z dot}=[inner sep=0mm, minimum size=2mm, shape=circle, draw=black, fill={rgb,255: red,221; green,255; blue,221}]
\tikzstyle{Z phase dot}=[minimum size=1.2em, font={\footnotesize\boldmath}, shape=rectangle, rounded corners=0.5em, inner sep=0.2em, outer sep=-0.2em, scale=0.8, tikzit shape=circle, draw=black, fill={rgb,255: red,221; green,255; blue,221}, tikzit draw=blue]
\tikzstyle{X dot}=[Z dot, shape=circle, draw=black, fill={rgb,255: red,255; green,136; blue,136}]
\tikzstyle{X phase dot}=[Z phase dot, tikzit shape=circle, tikzit draw=blue, fill={rgb,255: red,255; green,136; blue,136}, font={\footnotesize\boldmath}]
\tikzstyle{hadamard}=[fill=yellow, draw=black, shape=rectangle, inner sep=0.6mm, minimum height=1.5mm, minimum width=1.5mm]
\tikzstyle{vertex}=[inner sep=0mm, minimum size=1mm, shape=circle, draw=black, fill=black]
\tikzstyle{vertex set}=[inner sep=0mm, minimum size=1mm, shape=circle, draw=black, fill=white, font={\footnotesize\boldmath}]
\tikzstyle{target}=[inner sep=0mm, minimum size=3mm, shape=circle, draw=black]

\tikzstyle{hadamard edge}=[-, dashed, dash pattern=on 2pt off 1.5pt, thick, draw={rgb,255: red,68; green,136; blue,255}]
\tikzstyle{brace edge}=[-, tikzit draw=blue, decorate, decoration={brace,amplitude=1mm,raise=-1mm}]
\tikzstyle{diredge}=[->]
\tikzstyle{dashed edge}=[-, dashed, dash pattern=on 2pt off 0.5pt, draw=black]

\newtheorem{theorem}{Theorem}[section]
\newtheorem{lemma}[theorem]{Lemma}
\newtheorem{proposition}[theorem]{Proposition}
\newtheorem{corollary}[theorem]{Corollary}
\newtheorem{observation}[theorem]{Observation}
\theoremstyle{definition}
\newtheorem{definition}[theorem]{Definition}
\newtheorem{example}[theorem]{Example}
\newtheorem{remark}[theorem]{Remark}

\newcommand{\symdiff}{\ensuremath{\bigtriangleup}}

\let\sse=\subseteq

\let\ld=\lambda
\def\abs#1{\left| #1 \right|}
\renewcommand{\t}[1]{\ensuremath{^{\otimes #1}}}

\newcommand{\phase}[2]{\begin{tikzpicture}[baseline=-0.1cm]
	\begin{pgfonlayer}{nodelayer}
		\node [style=none] (0) at (-0.4, 0) {};
		\node [style=none] (2) at (0.4, 0) {};
		\node [style=#1 phase dot] (5) at (0, 0) {#2};
	\end{pgfonlayer}
	\begin{pgfonlayer}{edgelayer}
		\draw (0.center) to (5);
		\draw (5) to (2.center);
	\end{pgfonlayer}
\end{tikzpicture}}
\newcommand{\had}{\begin{tikzpicture}[baseline=-0.1cm]
	\begin{pgfonlayer}{nodelayer}
		\node [style=none] (0) at (-0.4, 0) {};
		\node [style=none] (2) at (0.4, 0) {};
		\node [style=H gate] (5) at (0, 0) {};
	\end{pgfonlayer}
	\begin{pgfonlayer}{edgelayer}
		\draw (0.center) to (5);
		\draw (5) to (2.center);
	\end{pgfonlayer}
\end{tikzpicture}}

\title{Flow-preserving ZX-calculus Rewrite Rules for Optimisation and Obfuscation}

\author{Tommy McElvanney
\institute{School of Computer Science\\
University of Birmingham}
\email{txm639@student.bham.ac.uk}
\and
Miriam Backens
\institute{School of Computer Science\\
University of Birmingham}
\email{m.backens@cs.bham.ac.uk}
}
\def\titlerunning{Flow-preserving ZX-calculus Rewrite Rules for Optimisation and Obfuscation}
\def\authorrunning{T. McElvanney \& M. Backens}
\begin{document}
\maketitle

\maketitle

\begin{abstract}
 In the one-way model of measurement-based quantum computation (MBQC),
computation proceeds via measurements on a resource state. So-called
flow conditions ensure that the overall computation is deterministic in
a suitable sense, with Pauli flow being the most general of these.
Computations, represented as measurement patterns, may be rewritten to
optimise resource use and for other purposes. Such rewrites need to
preserve the existence of flow to ensure the new pattern can still be
implemented deterministically. The majority of existing work in this
area has focused on rewrites that reduce the number of qubits, yet it
can be beneficial to increase the number of qubits for certain kinds of
optimisation, as well as for obfuscation.

In this work, we introduce several ZX-calculus rewrite rules that
increase the number of qubits and preserve the existence of Pauli flow.
These rules can be used to transform any measurement pattern into a
pattern containing only (general or Pauli) measurements within the
XY-plane. We also give the first flow-preserving rewrite rule that
allows measurement angles to be changed arbitrarily, and use this to
prove that the `neighbour unfusion' rule of Staudacher et al. preserves
the existence of Pauli flow. This implies it may be possible to reduce
the runtime of their two-qubit-gate optimisation procedure by removing
the need to regularly run the costly gflow-finding algorithm.
\end{abstract}

\section{Introduction}

The ZX-calculus is a graphical language for representing and reasoning about quantum computations, allowing us to conveniently represent and reason about computations in both the quantum circuit model and the one-way model of measurement based quantum-computation (MBQC), as well as to translate between the two. The ZX-calculus has various complete sets of rewrite rules, meaning that any two diagrams representing the same linear map can be transformed into each other entirely graphically \cite{Backens_2014,Jeandel_2018,Ng_2017} and provide tools for uses in optimization \cite{korbinian_2022} \cite{Duncan_2020}, obfuscation \cite{Cao_2022} and other areas of research in quantum computing.

The one-way model of MBQC involves the implementation of quantum computations by performing successive adaptive single-qubit measurements on a resource state \cite{Raussendorf_One-Way_2001}, largely without using any unitary operations. This contrasts with the more commonly-used circuit model and has applications in server-client scenarios as well as for certain quantum error-correcting codes.

An MBQC computation is given as a \emph{pattern}, which specifies the resource state -- usually a graph state -- and a sequence of measurements of certain types \cite{danos_parsimonious_2005}.
As measurements are non-deterministic, future measurements need to be adapted depending on the outcomes of past measurements to obtain an overall deterministic computation.
Yet not every pattern can be implemented deterministically.
Sufficient (and in some cases necessary) criteria for determinism are given by the different kinds of \emph{flow}, which define a partial order on the measured qubits and give instructions for how to adapt the future computation if a measurement yields the undesired outcome \cite{Danos_2006,Browne_2007} (cf.\ Section~\ref{s:pauli-flow}).  

In addition to the applications mentioned above, the flexible structure of MBQC patterns is also useful as a theoretical tool.
For example, translations between circuits and MBQC patterns have been used to trade off circuit depth versus qubit number \cite{broadbent_parallelizing_2009} or to reduce the number of $T$-gates in a Clifford+T circuit \cite{kissinger_reducing_2020}.
When translating an MBQC pattern (back) into a circuit, it is important that the pattern still have flow, as circuit extraction algorithms rely on flow \cite{Danos_2006,miyazaki_analysis_2015,Duncan_2020,Backens_2021}.

ZX-calculus diagrams directly corresponding to MBQC-patters are said to be in MBQC-form. Many of the standard ZX-calculus rewrite rules do not preserve the MBQC-form structure nor the existence of a flow, which we often want to preserve, thus circuit optimisation using MBQC and the ZX-calculus relies on proofs that preserve both MBQC-form and flow \cite{Duncan_2020,Backens_2021}. Much of the previous work on this has focused on rewrite rules that maintain or reduce the number of qubits, which find direct application in T-count optimisation \cite{Duncan_2020}. Nevertheless, it is sometimes desirable to increase the number of qubits in an MBQC pattern while preserving the existence of flow, such as for more involved optimisation strategies \cite{staudacher_optimization_2021} or for obfuscation.

In this work we introduce several ZX-calculus rewrite rules that preserve the MBQC-form structure as well as Pauli flow \cite{Browne_2007}, alongside proofs of this preservation. These rules have various applications, such as being used in obfuscation techniques for blind quantum computation \cite{Cao_2022}. Notably, we introduce the first Pauli flow preserving rewrite rule that allows us to change measurement angles arbitrarily, with all previous rules only allowing for changes that are integer multiples of $\frac{\pi}{2}$. Using this, we prove that the `neighbour unfusion’ rule of \cite{korbinian_2022} always preserves the existence of Pauli flow. Additionally, we provide a sufficient and a necessary condition for neighbour unfusion to preserve the existence of gflow \cite{Browne_2007}, a more restricted flow condition.

\section{Preliminaries}
\label{s:preliminaries}

In this section, we give an overview of the ZX-calculus and then use it to introduce measurement-based quantum computing.
We discuss the notion of flow that will be used in this paper and some existing rewrite rules which preserve the existence of this flow.

\subsection{The ZX-calculus}
The ZX-calculus is a diagrammatic language for reasoning about quantum computations. We will provide a short introduction here; for a more thorough overview, see \cite{VDW_2020, Coecke_2017}.

A ZX-diagram consists of \textit{spiders} and \textit{wires}. Diagrams are read from left to right: wires entering a diagram from the left are inputs while wires exiting the diagram on the right are outputs, like in the quantum circuit model. ZX-diagrams compose in
two distinct ways: \textit{horizontal composition}, which involves connecting the output wires of one diagram to the input wires of another, and \textit{vertical composition} (or the tensor product), which just involves drawing one diagram vertically above the other.
The linear map corresponding to a ZX-diagram $D$ is denoted by $\intf{D}$.

ZX-diagrams are generated by two families of spiders which may have any number of inputs or outputs, corresponding to the Z and X bases respectively. $Z$-spiders are drawn as green dots and $X$-spiders as red dots; with $m$ inputs, $n$ outputs, and using $(\cdot)\t{k}$ to denote a $k$-fold tensor power, we have:
\[
 \intf{\begin{tikzpicture}[baseline=-0.1cm, scale=0.5]
	\begin{pgfonlayer}{nodelayer}
		\node [style=Z phase dot] (0) at (0, 0) {$\alpha$};
		\node [style=none] (1) at (-1, 1) {};
		\node [style=none] (2) at (-1, 0.5) {};
		\node [style=none] (3) at (-1, -1) {};
		\node [style=none] (4) at (1, -1) {};
		\node [style=none] (5) at (1, 0.5) {};
		\node [style=none] (6) at (1, 1) {};
		\node [style=none] (7) at (-0.85, -0.15) {\vdots};
		\node [style=none] (8) at (0.85, -0.15) {\vdots};
	\end{pgfonlayer}
	\begin{pgfonlayer}{edgelayer}
		\draw (0) to (2.center);
		\draw [bend right=15] (0) to (1.center);
		\draw [bend left=15] (0) to (3.center);
		\draw [bend right=15] (0) to (4.center);
		\draw (0) to (5.center);
		\draw [bend left=15] (0) to (6.center);
	\end{pgfonlayer}
\end{tikzpicture}}
 = \ket{0}\t{n}\bra{0}\t{m} + e^{i\alpha}\ket{1}\t{n}\bra{1}\t{m} \qquad\qquad
 \intf{\begin{tikzpicture}[baseline=-0.1cm, scale=0.5]
	\begin{pgfonlayer}{nodelayer}
		\node [style=X phase dot] (0) at (0, 0) {$\alpha$};
		\node [style=none] (1) at (-1, 1) {};
		\node [style=none] (2) at (-1, 0.5) {};
		\node [style=none] (3) at (-1, -1) {};
		\node [style=none] (4) at (1, -1) {};
		\node [style=none] (5) at (1, 0.5) {};
		\node [style=none] (6) at (1, 1) {};
		\node [style=none] (7) at (-0.85, -0.15) {\vdots};
		\node [style=none] (8) at (0.85, -0.15) {\vdots};
	\end{pgfonlayer}
	\begin{pgfonlayer}{edgelayer}
		\draw (0) to (2.center);
		\draw [bend right=15] (0) to (1.center);
		\draw [bend left=15] (0) to (3.center);
		\draw [bend right=15] (0) to (4.center);
		\draw (0) to (5.center);
		\draw [bend left=15] (0) to (6.center);
	\end{pgfonlayer}
\end{tikzpicture}}
 = \ket{+}\t{n}\bra{+}\t{m} + e^{i\alpha}\ket{-}\t{n}\bra{-}\t{m}
\]

Spiders with exactly one input and output are unitary, in particular $\intf{\phase{Z}{$\alpha$}} = \ket{0}\bra{0} +e^{i\alpha}\ket{1}\bra{1} = Z_\alpha$ and $\intf{\phase{X}{$\alpha$}} = \ket{+}\bra{+} +e^{i\alpha}\ket{-}\bra{-} = X_\alpha$.

Two diagrams $D$ and $D'$ are said to be equivalent (written $D\cong D'$) if $\intf{D} = z \intf{D'}$ for some non-zero complex number $z$. For the rest of the paper, whenever we write a diagram equality we will mean equality up to some global scalar in this way.
For treatments of the ZX-calculus which do not ignore scalars see \cite{Backens_2015} for the stabilizer fragment, \cite{Jeandel_2018} for the Clifford+T fragment and \cite{Jeandel_2018_2, Ng_2017} for the full ZX-calculus.

The Hadamard gate $H=\ket{+}\bra{0}+\ket{-}\bra{1} \cong Z_{\frac{\pi}{2}} \circ X_{\frac{\pi}{2}} \circ Z_{\frac{\pi}{2}}$ will be widely used throughout the paper.
It has two common syntactic sugars -- a yellow square, or a blue dotted line -- with the latter only used between spiders:
\begin{center}
\begin{tikzpicture}[baseline=-0.1cm]
	\begin{pgfonlayer}{nodelayer}
		\node [style=none] (0) at (-2.75, 0) {};
		\node [style=none] (1) at (-1.75, 0) {};
		\node [style=none] (3) at (-0.75, 0) {};
		\node [style=none] (4) at (1.75, 0) {};
		\node [style=H gate] (5) at (-2.25, 0) {};
		\node [style=Z phase dot] (6) at (-0.25, 0) {$\frac{\pi}{2}$};
		\node [style=X phase dot] (7) at (0.5, 0) {$\frac{\pi}{2}$};
		\node [style=Z phase dot] (8) at (1.25, 0) {$\frac{\pi}{2}$};
		\node [style=none] (9) at (-1.25, 0) {$=$};
	\end{pgfonlayer}
	\begin{pgfonlayer}{edgelayer}
		\draw (0.center) to (5);
		\draw (5) to (1.center);
		\draw (3.center) to (6);
		\draw (6) to (7);
		\draw (7) to (8);
		\draw (8) to (4.center);
	\end{pgfonlayer}
\end{tikzpicture} \qquad\qquad \begin{tikzpicture}[baseline=-0.1cm]
	\begin{pgfonlayer}{nodelayer}
		\node [style=none] (0) at (0.5, 0) {};
		\node [style=Z dot] (1) at (1, 0) {};
		\node [style=H gate] (2) at (1.75, 0) {};
		\node [style=Z dot] (3) at (2.5, 0) {};
		\node [style=none] (4) at (3, 0) {};
		\node [style=none] (5) at (0, 0) {$=$};
		\node [style=none] (6) at (-2.5, 0) {};
		\node [style=Z dot] (7) at (-2, 0) {};
		\node [style=Z dot] (8) at (-1, 0) {};
		\node [style=none] (9) at (-0.5, 0) {};
	\end{pgfonlayer}
	\begin{pgfonlayer}{edgelayer}
		\draw (0.center) to (1);
		\draw (1) to (2);
		\draw (2) to (3);
		\draw (3) to (4.center);
		\draw (6.center) to (7);
		\draw (8) to (9.center);
		\draw [style=Hadamard edge] (7) to (8);
	\end{pgfonlayer}
\end{tikzpicture}
\end{center}

The ZX-calculus is equipped with a set of rewrite rules which can be used to transform a ZX-diagram into another diagram representing the same linear map.
The following rules, along with the definition of \had{}, is complete for stabilizer ZX-diagrams: any two stabilizer ZX-diagrams which correspond (up to non-zero scalar factor) to the same linear map can be rewritten into one another using these rules \cite{Backens_2014}.

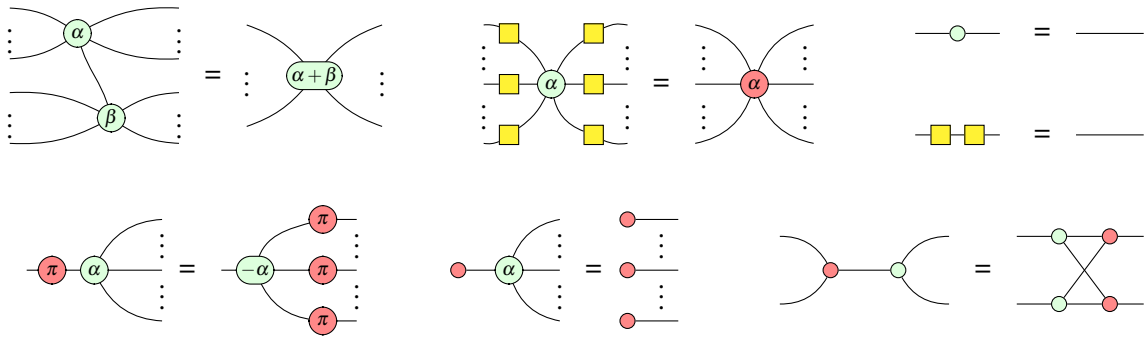
\begin{figure}
	\begin{tikzpicture}[scale=0.45]
	\begin{pgfonlayer}{nodelayer}
		\node [style=Z phase dot] (0) at (-17, 4.75) {$\alpha$};
		\node [style=Z phase dot] (1) at (-16, 2.25) {$\beta$};
		\node [style=none] (2) at (-19, 5.5) {};
		\node [style=none] (3) at (-19, 4) {};
		\node [style=none] (4) at (-19, 3) {};
		\node [style=none] (5) at (-19, 1.5) {};
		\node [style=none] (6) at (-14, 5.5) {};
		\node [style=none] (7) at (-14, 4) {};
		\node [style=none] (8) at (-14, 3) {};
		\node [style=none] (9) at (-14, 1.5) {};
		\node [style=none] (10) at (-19, 4.75) {$\vdots$};
		\node [style=none] (11) at (-19, 2.25) {$\vdots$};
		\node [style=none] (12) at (-14, 2.25) {$\vdots$};
		\node [style=none] (13) at (-14, 4.75) {$\vdots$};
		\node [style=none] (14) at (-13, 3.5) {=};
		\node [style=Z phase dot] (15) at (-10, 3.5) {$\alpha + \beta$};
		\node [style=none] (16) at (-12, 5) {};
		\node [style=none] (17) at (-12, 2) {};
		\node [style=none] (18) at (-8, 5) {};
		\node [style=none] (19) at (-8, 2) {};
		\node [style=none] (20) at (-12, 3.5) {$\vdots$};
		\node [style=none] (21) at (-8, 3.5) {$\vdots$};
		\node [style=Z phase dot] (23) at (-3, 3.25) {$\alpha$};
		\node [style=none] (24) at (-5, 5) {};
		\node [style=none] (25) at (-5, 3.25) {};
		\node [style=none] (26) at (-5, 1.5) {};
		\node [style=none] (27) at (-0.75, 5) {};
		\node [style=none] (28) at (-0.75, 3.25) {};
		\node [style=none] (29) at (-0.75, 1.5) {};
		\node [style=H gate] (30) at (-4.25, 4.75) {};
		\node [style=H gate] (31) at (-4.25, 3.25) {};
		\node [style=H gate] (32) at (-4.25, 1.75) {};
		\node [style=H gate] (33) at (-1.75, 1.75) {};
		\node [style=H gate] (34) at (-1.75, 3.25) {};
		\node [style=H gate] (35) at (-1.75, 4.75) {};
		\node [style=none] (36) at (0.25, 3.25) {=};
		\node [style=X phase dot] (38) at (3, 3.25) {$\alpha$};
		\node [style=none] (39) at (1.25, 3.25) {};
		\node [style=none] (40) at (1.25, 5) {};
		\node [style=none] (41) at (1.25, 1.5) {};
		\node [style=none] (42) at (4.75, 5) {};
		\node [style=none] (43) at (4.75, 3.25) {};
		\node [style=none] (44) at (4.75, 1.5) {};
		\node [style=none] (45) at (7.75, 4.75) {};
		\node [style=none] (46) at (10.25, 4.75) {};
		\node [style=none] (47) at (12.5, 4.75) {};
		\node [style=none] (48) at (14.5, 4.75) {};
		\node [style=none] (49) at (11.5, 4.75) {=};
		\node [style=Z dot] (51) at (9, 4.75) {};
		\node [style=none] (52) at (7.75, 1.75) {};
		\node [style=none] (53) at (10.25, 1.75) {};
		\node [style=none] (54) at (12.5, 1.75) {};
		\node [style=none] (55) at (14.5, 1.75) {};
		\node [style=none] (56) at (11.5, 1.75) {=};
		\node [style=H gate] (58) at (8.5, 1.75) {};
		\node [style=H gate] (59) at (9.5, 1.75) {};
		\node [style=none] (60) at (-18.5, -2.25) {};
		\node [style=X phase dot] (61) at (-17.75, -2.25) {$\pi$};
		\node [style=Z phase dot] (62) at (-16.5, -2.25) {$\alpha$};
		\node [style=none] (63) at (-14.5, -0.75) {};
		\node [style=none] (64) at (-14.5, -2.25) {};
		\node [style=none] (65) at (-14.5, -3.75) {};
		\node [style=none] (66) at (-14.5, -1.3) {$\vdots$};
		\node [style=none] (67) at (-14.5, -2.8) {$\vdots$};
		\node [style=none] (68) at (-12.75, -2.25) {};
		\node [style=Z phase dot] (70) at (-11.75, -2.25) {$-\alpha$};
		\node [style=none] (71) at (-8.75, -0.75) {};
		\node [style=none] (72) at (-8.75, -2.25) {};
		\node [style=none] (73) at (-8.75, -3.75) {};
		\node [style=none] (74) at (-8.75, -1.3) {$\vdots$};
		\node [style=none] (75) at (-8.75, -2.8) {$\vdots$};
		\node [style=none] (76) at (-13.75, -2.25) {=};
		\node [style=X phase dot] (77) at (-9.75, -0.75) {$\pi$};
		\node [style=X phase dot] (78) at (-9.75, -2.25) {$\pi$};
		\node [style=X phase dot] (79) at (-9.75, -3.75) {$\pi$};
		\node [style=X dot] (82) at (-5.75, -2.25) {};
		\node [style=Z phase dot] (83) at (-4.25, -2.25) {$\alpha$};
		\node [style=none] (84) at (-2.75, -0.75) {};
		\node [style=none] (85) at (-2.75, -2.25) {};
		\node [style=none] (86) at (-2.75, -3.75) {};
		\node [style=none] (87) at (-2.75, -1.3) {$\vdots$};
		\node [style=none] (88) at (-2.75, -2.8) {$\vdots$};
		\node [style=none] (89) at (-1.75, -2.25) {=};
		\node [style=X dot] (91) at (-0.75, -0.75) {};
		\node [style=X dot] (92) at (-0.75, -2.25) {};
		\node [style=X dot] (93) at (-0.75, -3.75) {};
		\node [style=none] (94) at (0.75, -0.75) {};
		\node [style=none] (95) at (0.75, -2.25) {};
		\node [style=none] (96) at (0.75, -3.75) {};
		\node [style=none] (97) at (3.75, -1.25) {};
		\node [style=none] (98) at (3.75, -3.25) {};
		\node [style=X dot] (99) at (5.25, -2.25) {};
		\node [style=Z dot] (100) at (7.25, -2.25) {};
		\node [style=none] (101) at (8.75, -1.25) {};
		\node [style=none] (102) at (8.75, -3.25) {};
		\node [style=none] (103) at (10.75, -1.25) {};
		\node [style=none] (104) at (10.75, -3.25) {};
		\node [style=none] (107) at (14.5, -1.25) {};
		\node [style=none] (108) at (14.5, -3.25) {};
		\node [style=Z dot] (109) at (12, -1.25) {};
		\node [style=Z dot] (110) at (12, -3.25) {};
		\node [style=X dot] (111) at (13.5, -1.25) {};
		\node [style=X dot] (112) at (13.5, -3.25) {};
		\node [style=none] (113) at (9.75, -2.25) {=};
		\node [style=none] (115) at (0.25, -1.3) {$\vdots$};
		\node [style=none] (116) at (0.25, -2.8) {$\vdots$};
		\node [style=none] (125) at (-5, 4.25) {$\vdots$};
		\node [style=none] (126) at (-0.75, 4.25) {$\vdots$};
		\node [style=none] (127) at (-0.75, 2.5) {$\vdots$};
		\node [style=none] (129) at (1.5, 4.25) {$\vdots$};
		\node [style=none] (130) at (1.5, 2.5) {$\vdots$};
		\node [style=none] (131) at (4.5, 2.5) {$\vdots$};
		\node [style=none] (132) at (4.5, 4.25) {$\vdots$};
		\node [style=none] (133) at (-5, 2.5) {$\vdots$};
	\end{pgfonlayer}
	\begin{pgfonlayer}{edgelayer}
		\draw [bend right=15] (0) to (2.center);
		\draw [bend left=15] (0) to (3.center);
		\draw [bend right=15] (1) to (4.center);
		\draw [bend left=15] (1) to (5.center);
		\draw [in=105, out=-75] (0) to (1);
		\draw [bend left=15] (1) to (8.center);
		\draw [bend right=15] (1) to (9.center);
		\draw [bend right=15] (0) to (7.center);
		\draw [bend left=15] (0) to (6.center);
		\draw [bend right=15] (15) to (16.center);
		\draw [bend left=15] (15) to (17.center);
		\draw [bend left=15] (15) to (18.center);
		\draw [bend right=15] (15) to (19.center);
		\draw [bend right=15] (23) to (30);
		\draw [bend right=15, looseness=0.75] (30) to (24.center);
		\draw (23) to (31);
		\draw (31) to (25.center);
		\draw [bend left=15] (23) to (32);
		\draw [bend left=15] (32) to (26.center);
		\draw [bend right] (23) to (33);
		\draw [bend right=15] (33) to (29.center);
		\draw (23) to (34);
		\draw (34) to (28.center);
		\draw [bend left=15] (23) to (35);
		\draw [bend left=15, looseness=1.25] (35) to (27.center);
		\draw [bend right] (38) to (40.center);
		\draw (39.center) to (38);
		\draw [bend left] (38) to (41.center);
		\draw [bend right] (38) to (44.center);
		\draw (43.center) to (38);
		\draw [bend left] (38) to (42.center);
		\draw (45.center) to (51);
		\draw (51) to (46.center);
		\draw (47.center) to (48.center);
		\draw (54.center) to (55.center);
		\draw (52.center) to (58);
		\draw (58) to (59);
		\draw (59) to (53.center);
		\draw (60.center) to (61);
		\draw (61) to (62);
		\draw [bend left] (62) to (63.center);
		\draw (62) to (64.center);
		\draw [bend right] (62) to (65.center);
		\draw (68.center) to (70);
		\draw [in=-165, out=75] (70) to (77);
		\draw (70) to (78);
		\draw [bend right] (70) to (79);
		\draw (79) to (73.center);
		\draw (78) to (72.center);
		\draw (77) to (71.center);
		\draw (82) to (83);
		\draw [bend left] (83) to (84.center);
		\draw (83) to (85.center);
		\draw [bend right] (83) to (86.center);
		\draw (91) to (94.center);
		\draw (92) to (95.center);
		\draw (93) to (96.center);
		\draw (99) to (100);
		\draw [bend right] (99) to (97.center);
		\draw [bend left] (99) to (98.center);
		\draw [bend left] (100) to (101.center);
		\draw [bend right] (100) to (102.center);
		\draw (103.center) to (109);
		\draw (111) to (107.center);
		\draw (112) to (108.center);
		\draw (110) to (104.center);
		\draw (109) to (112);
		\draw (110) to (111);
		\draw (109) to (111);
		\draw (110) to (112);
	\end{pgfonlayer}
\end{tikzpicture}
	\caption{A complete set of rewrite rules for the scalar-free stabilizer ZX-calculus. Each rule also holds with the colours or the directions reversed.}
	\label{fig:ZX-rules}
\end{figure}

\subsection{Measurement-based Quantum computation}

Measurement-based Quantum computation (MBQC) is a particularly interesting model of quantum computation with no classical analogue. In the one-way model of MBQC, one first constructs a highly entangled resource state that can be independent of the specific computation that one wants to perform (only depending on the `size' of the computation) by preparing qubits in the $\ket{+}$ state and applying $CZ$-gates to certain pairs of qubits. The computation then proceeds by performing single qubit measurements in a specified order. MBQC is a universal model for quantum computation -- any computation can be performed by choosing an appropriate resource state and then performing a certain combination of measurements on said state.

Measurement-based computations are traditionally expressed as \textit{measurement patterns}, which use a sequence of commands to describe how the resource state is constructed and how the computation proceeds \cite{danos_parsimonious_2005}.
As the resource states are graph states, a graphical representation of MBQC protocols can be more intuitive; we shall therefore introduce MBQC with ZX-diagrams.

\begin{definition}[\cite{Duncan_2009}]
A \emph{graph state diagram} is a ZX-diagram where each vertex is a (phase-free) green spider,
each edge connecting spiders has a Hadamard gate on it, and
there is a single output wire incident on each vertex.
\end{definition}

\begin{definition}\cite[Definition~2.18]{Backens_2021}
\label{def:MBQC-form}
A ZX-diagram is in \emph{MBQC-form} if it consists of a
graph state diagram in which each vertex of the graph may furthermore be connected to an input (in addition to its output), and a measurement effect instead of its output.
\end{definition}

MBQC restricts the allowed single-qubit measurements to three planes of the Bloch sphere: those spanned by the eigenstates of two Pauli matrices, called the XY, YZ and XZ planes. Each time a qubit $u$ is measured in a plane $\lambda(u)$ at an angle $\alpha$, one may obtain either the desired outcome, denoted $\bra{+_{\lambda(u),\alpha}}$, or the undesired outcome $\bra{-_{\lambda(u), \alpha}} = \bra{+_{\lambda(u), \alpha+\pi}}$.
Measurements where the angle is an integer multiple of $\frac{\pi}{2}$ are Pauli measurements; the corresponding measurement type is denoted by simply $X$, $Y$, or $Z$.
The ZX-diagram corresponding to each (desired) measurement outcome is given in Table~\ref{tab:MBQC-ZX}.
The structure of an MBQC protocol is formalised as follows.

\begin{table}
	\centering	
	\renewcommand{\arraystretch}{1.5}
	\begin{tabular}{c||c|c|c|c|c|c|c|c|c|}
		operator & $\bra{+_{XY, \alpha}}_i$ & $\bra{+_{XZ, \alpha}}_i$ & $\bra{+_{YZ, \alpha}}_i$ & $\bra{+_{X,0}}_i$ & $\bra{+_{Y,0}}_i$ & $\bra{+_{Z,0}}_i$ & $\bra{+_{X,\pi}}_i$ & $\bra{+_{Y,\pi}}_i$ & $\bra{+_{Z,\pi}}_i$ \\ \hline
		diagram &  \begin{tikzpicture}[baseline=-0.05]
	\begin{pgfonlayer}{nodelayer}
		\node [style=none] (0) at (-0.25, 0) {};
		\node [style=Z phase dot] (1) at (0.25, 0) {$-\alpha(i)$};
	\end{pgfonlayer}
	\begin{pgfonlayer}{edgelayer}
		\draw (0.center) to (1);
	\end{pgfonlayer}
\end{tikzpicture}  &  \begin{tikzpicture}[baseline=-0.05, scale=.8]
	\begin{pgfonlayer}{nodelayer}
		\node [style=none] (0) at (-0.25, 0) {};
		\node [style=Z phase dot] (1) at (0.25, 0) {$\frac{\pi}{2}$};
		\node [style=X phase dot] (2) at (1, 0) {$\alpha(i)$};
	\end{pgfonlayer}
	\begin{pgfonlayer}{edgelayer}
		\draw (0.center) to (1);
		\draw (1) to (2);
	\end{pgfonlayer}
\end{tikzpicture} &  \begin{tikzpicture}[baseline=-0.05]
	\begin{pgfonlayer}{nodelayer}
		\node [style=none] (0) at (-0.25, 0) {};
		\node [style=X phase dot] (1) at (0.25, 0) {$\alpha(i)$};
	\end{pgfonlayer}
	\begin{pgfonlayer}{edgelayer}
		\draw (0.center) to (1);
	\end{pgfonlayer}
\end{tikzpicture} &  \begin{tikzpicture}[baseline=-0.05]
	\begin{pgfonlayer}{nodelayer}
		\node [style=none] (0) at (-0.25, 0) {};
		\node [style=Z dot] (1) at (0.25, 0) {};
	\end{pgfonlayer}
	\begin{pgfonlayer}{edgelayer}
		\draw (0.center) to (1);
	\end{pgfonlayer}
\end{tikzpicture}  &  \begin{tikzpicture}[baseline=-0.05]
	\begin{pgfonlayer}{nodelayer}
		\node [style=none] (0) at (-0.25, 0) {};
		\node [style=Z phase dot] (1) at (0.25, 0) {-$\frac{\pi}{2}$};
	\end{pgfonlayer}
	\begin{pgfonlayer}{edgelayer}
		\draw (0.center) to (1);
	\end{pgfonlayer}
\end{tikzpicture} &  \begin{tikzpicture}[baseline=-0.05]
	\begin{pgfonlayer}{nodelayer}
		\node [style=none] (0) at (-0.25, 0) {};
		\node [style=X dot] (1) at (0.25, 0) {};
	\end{pgfonlayer}
	\begin{pgfonlayer}{edgelayer}
		\draw (0.center) to (1);
	\end{pgfonlayer}
\end{tikzpicture} &  \begin{tikzpicture}[baseline=-0.05]
	\begin{pgfonlayer}{nodelayer}
		\node [style=none] (0) at (-0.25, 0) {};
		\node [style=Z phase dot] (1) at (0.25, 0) {$\pi$};
	\end{pgfonlayer}
	\begin{pgfonlayer}{edgelayer}
		\draw (1) to (0.center);
	\end{pgfonlayer}
\end{tikzpicture} &  \begin{tikzpicture}[baseline=-0.05]
	\begin{pgfonlayer}{nodelayer}
		\node [style=none] (0) at (-0.25, 0) {};
		\node [style=Z phase dot] (1) at (0.25, 0) {$\frac{\pi}{2}$};
	\end{pgfonlayer}
	\begin{pgfonlayer}{edgelayer}
		\draw (1) to (0.center);
	\end{pgfonlayer}
\end{tikzpicture} &  \begin{tikzpicture}[baseline=-0.05]
	\begin{pgfonlayer}{nodelayer}
		\node [style=none] (0) at (-0.25, 0) {};
		\node [style=X phase dot] (1) at (0.25, 0) {$\pi$};
	\end{pgfonlayer}
	\begin{pgfonlayer}{edgelayer}
		\draw (1) to (0.center);
	\end{pgfonlayer}
\end{tikzpicture}\\ 
	\end{tabular}
	\renewcommand{\arraystretch}{1}
	\caption{MBQC measurement effects in Dirac notation and their corresponding ZX-diagrams}
	\label{tab:MBQC-ZX}
\end{table}

\begin{definition} 	
A \emph{labelled open graph} is a tuple $\Gamma = (G, I, O, \lambda)$,
where $G=(V,E)$ is a simple undirected graph, $I\subseteq V$ is a set of input vertices, $O \subseteq V$ is a set of output vertices, and $\lambda: V\setminus O \to \{X, Y, Z, \XYm, \XZm, \YZm\}$ assigns a measurement plane or Pauli measurement to each non-output vertex.
\end{definition}

\subsection{Pauli flow}
\label{s:pauli-flow}

Measurement-based computations are inherently probabilistic because measurements are probabilistic.
Computations can be made deterministic overall (up to Pauli corrections on the outputs) by tracking which measurements result in undesired outcomes and then correcting for these by adapting future measurements.
A sufficient (and in some cases necessary) condition for this to be possible on a given labelled open graph is \emph{Pauli flow}.
In the following, $\mathcal{P}(A)$ denotes the powerset of a set $A$, $N_G(v) = \{w \in V \mid (v,w) \in E\}$ is the set of neighbours of a vertex $v$ in a graph $G = (V,E)$. Furthermore, $\odds{G}{A} = \{v \in V \mid \abs{N_G(v) \cap A} \equiv 1 \bmod 2\}$ is the set of vertices in the graph $G=(V,E)$ which have an odd number of neighbours in $A\subseteq V$; this is referred to as the odd neighbourhood (in $G$) of $A$.

\begin{definition}[{\cite[Definition 5]{Browne_2007}}]\label{def:Pauli-flow}
A labelled open graph $(G, I, O, \lambda)$ has Pauli flow if there exists a map $p: V\setminus O \to \mathcal{P}(V\setminus I)$ and a partial order $\prec$ over V such that for all $u \in V\setminus O$,
\begin{enumerate}
    \item if $v\in p(u)$, $v \not = u$ and $\lambda(v)\not \in \{X,Y\}$, then $u \prec v$. \label{P1}
    \item if $v\in \mathrm{Odd}_G(p(u))$, $v \not = u$ and $\lambda(v)\not \in \{Y, Z\}$, then $u \prec v$. \label{P2}
    \item if $\lnot (u \prec v)$, $u \not = v$	 and $\lambda(v) = Y$, then $v \in p(u) \Longleftrightarrow v \in \mathrm{Odd}_G(p(u))$. \label{P3}
    \item if $\lambda(u) = XY$, then $u\not \in p(u)$ and $u \in \mathrm{Odd}_G(p(u))$. \label{P4}
    \item if $\lambda(u) = XZ$, then $u \in p(u)$ and $u \in \mathrm{Odd}_G(p(u))$. \label{P5}
    \item if $\lambda(u) = YZ$, then $u\in p(u)$ and $u \not \in \mathrm{Odd}_G(p(u))$.\label{P6}
    \item if $\lambda(u) = X$, then $u \in \mathrm{Odd}_G(p(u))$. \label{P7}
    \item if $\lambda(u) = Z$, then $u \in p(u)$. \label{P8}
    \item if $\lambda(u) = Y$ then either $u \in p(u)$ and $u\not \in \mathrm{Odd}_G(p(u))$ or $u \not \in p(u)$ and $u \in \mathrm{Odd}_G(p(u))$. \label{P9}
\end{enumerate}
\end{definition}

Here, the partial order is related to time order in which the qubits need to be measured.
The set $p(u)$ denotes qubits that are modified by Pauli-$X$ to compensate for an undesired measurement outcome on $u$, $\mathrm{Odd}_G(p(u))$ denotes the set of vertices that are modified by Pauli-$Z$.

Pauli flow is a sufficient condition for strong, stepwise and uniform determinism: this means all branches of the computation should implement the same linear operator up to a phase, any interval of the computation should be deterministic on its own, and the computation should be deterministic for all choices of measurement angles that satisfy $\ld$ \cite[p.~5]{Browne_2007}.
Pauli flow (and related flow conditions) are particularly interesting from a ZX-calculus perspective as there are polynomial-time algorithms for extracting circuits from MBQC-form ZX-diagrams with flow \cite{Duncan_2020,Backens_2021,Simmons_2021}, while circuit extraction from general ZX-diagrams is \#P-hard \cite{debeaudrap_2022}.

There are certain Pauli flows where the side effects of any correction are particularly limited, these are called \emph{focused} and they exist whenever a labelled open graph has Pauli flow.

\begin{definition}[{rephrased from \cite[Definition~4.3]{Simmons_2021}}]\label{def:focused_Pauli-flow}
 Suppose the labelled open graph $(G, I, O, \lambda)$ has a Pauli flow $(p,\prec)$.
 Define $S_u = V\setminus(O\cup\{u\})$ for all $u\in V$.
 Then $(p,\prec)$ is \emph{focused} if for all $u\in V\setminus O$:
 \begin{itemize}
  \item Any $v\in S_u\cap p(u)$ satisfies $\ld(v)\in\{\XYm, X,Y\}$.
  \item Any $v\in S_u\cap \odd{p(u)}$ satisfies $\ld(v)\in\{\XZm, \YZm, Y, Z\}$.
  \item Any $v\in S_u$ such that $\ld(v)=Y$ satisfies $v\in p(u)$ if and only if $v\in\odd{p(u)}$.
 \end{itemize}
\end{definition}

\begin{lemma}[{\cite[Lemma~4.6]{Simmons_2021}}]
 If a labelled open graph has Pauli flow, then it has a focused Pauli flow.
\end{lemma}

If the codomain of $\ld$ in Definition~\ref{def:Pauli-flow} is $\{\XYm,\XZm,\YZm\}$, the flow is called a gflow \cite{Browne_2007}.
Similarly, for Definition~\ref{def:focused_Pauli-flow}, the restriction to measurement labels $\{\XYm,\XZm,\YZm\}$ is called a focused gflow \cite{Backens_2021}.
If a labelled open graph has gflow, then it has a focused gflow \cite[Proposition~3.14]{Backens_2021}.

\subsection{Existing flow-preserving rewrite rules}
\label{s:existing-rules}

The basic ZX-calculus rewrite rules in Figure~\ref{fig:ZX-rules} do not generally preserve even the MBQC-form structure of a ZX-calculus diagram.
Yet there are some more complex derived rewrite rules that are known to preserve both the MBQC-form structure and the existence of a flow.
These rules were previously considered in the context of gflow \cite{Duncan_2020} and extended gflow \cite{Backens_2021}; the Pauli-flow preservation proofs are due to \cite{Simmons_2021}.
The simplest of these rules are $Z$-deletion and $Z$-insertion:

\begin{lemma}[{\cite[Lemma D.6]{Simmons_2021}}]
\label{lem:Z-delete}
Deleting a $Z$-measured vertex preserves the existence of Pauli flow.
\begin{center}
    \begin{tikzpicture}[scale=0.5]
	\begin{pgfonlayer}{nodelayer}
		\node [style=Z dot] (0) at (-3.25, 0) {};
		\node [style=Z dot] (1) at (-4.25, 1.25) {};
		\node [style=Z dot] (2) at (-4.25, -1.25) {};
		\node [style=Z dot] (3) at (-2.25, -1.25) {};
		\node [style=Z dot] (4) at (-2.25, 1.25) {};
		\node [style=none] (5) at (-1, 2) {};
		\node [style=none] (6) at (-1, 0.5) {};
		\node [style=none] (7) at (-1, -0.5) {};
		\node [style=none] (8) at (-1, -2) {};
		\node [style=none] (9) at (-5.75, 2) {};
		\node [style=none] (10) at (-5.75, 0.5) {};
		\node [style=none] (11) at (-5.75, -0.5) {};
		\node [style=none] (12) at (-5.75, -2) {};
		\node [style=none] (13) at (-5.75, 1.25) {$\vdots$};
		\node [style=none] (15) at (-1, -1.25) {$\vdots$};
		\node [style=none] (16) at (-1, 1.25) {$\vdots$};
		\node [style=none] (17) at (-5.75, -1.25) {$\vdots$};
		\node [style=X phase dot] (18) at (-3.25, 1.5) {$a\pi$};
		\node [style=none] (19) at (0, 0) {=};
		\node [style=none] (37) at (-4.25, 0.25) {$\vdots$};
		\node [style=none] (38) at (-2.25, 0.25) {$\vdots$};
		\node [style=Z phase dot] (42) at (2.25, 1.25) {$a\pi$};
		\node [style=Z phase dot] (43) at (2.25, -1.25) {$a\pi$};
		\node [style=Z phase dot] (44) at (4.25, -1.25) {$a\pi$};
		\node [style=Z phase dot] (45) at (4.25, 1.25) {$a\pi$};
		\node [style=none] (46) at (5.5, 2) {};
		\node [style=none] (47) at (5.5, 0.5) {};
		\node [style=none] (48) at (5.5, -0.5) {};
		\node [style=none] (49) at (5.5, -2) {};
		\node [style=none] (50) at (0.75, 2) {};
		\node [style=none] (51) at (0.75, 0.5) {};
		\node [style=none] (52) at (0.75, -0.5) {};
		\node [style=none] (53) at (0.75, -2) {};
		\node [style=none] (54) at (0.75, 1.25) {$\vdots$};
		\node [style=none] (55) at (5.5, -1.25) {$\vdots$};
		\node [style=none] (56) at (5.5, 1.25) {$\vdots$};
		\node [style=none] (57) at (0.75, -1.25) {$\vdots$};
		\node [style=none] (59) at (2.25, 0.25) {$\vdots$};
		\node [style=none] (60) at (4.25, 0.25) {$\vdots$};
	\end{pgfonlayer}
	\begin{pgfonlayer}{edgelayer}
		\draw (1) to (9.center);
		\draw (1) to (10.center);
		\draw (2) to (11.center);
		\draw (2) to (12.center);
		\draw (3) to (8.center);
		\draw (3) to (7.center);
		\draw (4) to (6.center);
		\draw (4) to (5.center);
		\draw (18) to (0);
		\draw [style=Hadamard edge] (1) to (0);
		\draw [style=Hadamard edge] (0) to (4);
		\draw [style=Hadamard edge] (0) to (3);
		\draw [style=Hadamard edge] (0) to (2);
		\draw (42) to (50.center);
		\draw (42) to (51.center);
		\draw (43) to (52.center);
		\draw (43) to (53.center);
		\draw (44) to (49.center);
		\draw (44) to (48.center);
		\draw (45) to (47.center);
		\draw (45) to (46.center);
	\end{pgfonlayer}
\end{tikzpicture}
\end{center}
\end{lemma}

\begin{lemma}[{\cite[Proposition 4.1]{McElvanney_2022}}]\label{lem:Z-insert}
	Inserting a $Z$-measured vertex (i.e.\ the inverse of $Z$-deletion) also preserves the existence of Pauli flow.
\end{lemma}

Other rewrite rules are based around quantum generalisations of two graph-theoretic operations.

\begin{definition}
Let $G=(V,E)$ be a graph and $u\in V$. The \emph{local complementation of $G$ about $u$} is the operation which maps $G$ to $G\star u := (V, E\symdiff \{(b,c) \mid (b,u),(c,u) \in E \text{ and } b\neq c\})$, where $\symdiff$ is the symmetric difference operator given by $A\symdiff B = (A\cup B)\setminus(A\cap B)$.
 The \emph{pivot of $G$ about the edge $(u,v)$} is the operation mapping $G$ to the graph $G \wedge uv := G\star u \star v \star u$.
\end{definition}

Local complementation keeps the vertices of the graph the same but toggles some edges: for each pair of neighbours of $u$, i.e. $v, v' \in N_G(u)$, there is an edge connecting $v$ and $v'$ in $G\star u$ if and only if there is no edge connecting $v$ and $v'$ in $G$.
Pivoting is a series of three local complementations about two neighbouring vertices, and is denoted by $G\wedge uv= G \star u \star v \star u$.

Both local complementation and pivoting give rise to operations on MBQC-form diagrams which preserve the MBQC form as well as the existence of Pauli flow (after some simple merging of single-qubit Cliffords into measurement effects, cf.~\cite[Section~4.2]{Backens_2021}).
We illustrate the operations with examples as they are difficult to express in ZX-calculus in generality.

\begin{lemma}[{\cite[Lemma D.12]{Simmons_2021}}]
 \label{lem:local-complementation}
 A local complementation about a vertex $u$ preserves the existence of Pauli flow.
 \begin{center}
      \begin{tikzpicture}[xscale=0.4,yscale=0.3]
	\begin{pgfonlayer}{nodelayer}
		\node [style=Z dot] (0) at (2, 3) {};
		\node [style=Z dot] (1) at (-1, 1) {};
		\node [style=Z dot] (2) at (1, -1) {};
		\node [style=Z dot] (3) at (-1, -3) {};
		\node [style=none] (4) at (4, 3) {};
		\node [style=none] (5) at (4, 1) {};
		\node [style=none] (6) at (4, -1) {};
		\node [style=none] (7) at (4, -3) {};
		\node [style=none] (8) at (1, -1.9) {$\vdots$};
		\node [style=none] (9) at (1, 3) {u};
		\node [style=none] (10) at (5, 0) {=};
	\end{pgfonlayer}
	\begin{pgfonlayer}{edgelayer}
		\draw [style=Hadamard edge] (0) to (3);
		\draw [style=Hadamard edge] (1) to (3);
		\draw [style=Hadamard edge] (1) to (0);
		\draw [style=Hadamard edge] (0) to (2);
		\draw (0) to (4.center);
		\draw (1) to (5.center);
		\draw (2) to (6.center);
		\draw (3) to (7.center);
	\end{pgfonlayer}
\end{tikzpicture}  \begin{tikzpicture}[xscale=0.4,yscale=0.3]
	\begin{pgfonlayer}{nodelayer}
		\node [style=Z dot] (0) at (0, 3) {};
		\node [style=Z dot] (1) at (-3, 1) {};
		\node [style=Z dot] (2) at (-1, -1) {};
		\node [style=Z dot] (3) at (-3, -3) {};
		\node [style=none] (4) at (4, 3) {};
		\node [style=none] (5) at (4, 1) {};
		\node [style=none] (6) at (4, -1) {};
		\node [style=none] (7) at (4, -3) {};
		\node [style=none] (8) at (-1, -1.9) {$\vdots$};
		\node [style=X phase dot] (9) at (2, 3) {$-\frac{\pi}{2}$};
		\node [style=Z phase dot] (10) at (2, 1) {$\frac{\pi}{2}$};
		\node [style=Z phase dot] (11) at (2, -1) {$\frac{\pi}{2}$};
		\node [style=Z phase dot] (12) at (2, -3) {$\frac{\pi}{2}$};
		\node [style=none] (15) at (-1, 3) {u};
	\end{pgfonlayer}
	\begin{pgfonlayer}{edgelayer}
		\draw [style=Hadamard edge] (0) to (3);
		\draw [style=Hadamard edge] (1) to (0);
		\draw [style=Hadamard edge] (0) to (2);
		\draw (0) to (9);
		\draw (9) to (4.center);
		\draw (1) to (10);
		\draw (10) to (5.center);
		\draw (2) to (11);
		\draw (11) to (6.center);
		\draw (3) to (12);
		\draw (12) to (7.center);
		\draw [style=Hadamard edge] (1) to (2);
		\draw [style=Hadamard edge] (2) to (3);
	\end{pgfonlayer}
\end{tikzpicture}
 \end{center}
\end{lemma}

\begin{lemma}[{\cite[Lemma D.21]{Simmons_2021}}]
 \label{lem:pivot}
 A pivot about an edge $(u,v)$ preserves the existence of Pauli flow.
\begin{center}
    \begin{tikzpicture}[scale=0.4]
	\begin{pgfonlayer}{nodelayer}
		\node [style=Z dot] (0) at (2, 3) {};
		\node [style=Z dot] (1) at (2, -3) {};
		\node [style=Z dot] (2) at (0, -1) {};
		\node [style=Z dot] (3) at (0, 1) {};
		\node [style=Z dot] (4) at (-2.5, 3.25) {};
		\node [style=Z dot] (5) at (-4, 2) {};
		\node [style=Z dot] (6) at (-4, -2) {};
		\node [style=Z dot] (7) at (-2.5, -3.25) {};
		\node [style=none] (8) at (-4.5, -2.5) {};
		\node [style=none] (9) at (-3, -3.75) {};
		\node [style=none] (10) at (-4.5, 2.5) {};
		\node [style=none] (11) at (-3, 3.75) {};
		\node [style=none] (12) at (5, 1) {};
		\node [style=none] (13) at (5, -1) {};
		\node [style=none] (14) at (5, 3) {};
		\node [style=none] (15) at (5, -3.25) {};
		\node [style=none] (16) at (6, 0) {=};
		\node [style=none] (17) at (2.5, 3.5) {$v$};
		\node [style=none] (18) at (2.5, -3.5) {$u$};
	\end{pgfonlayer}
	\begin{pgfonlayer}{edgelayer}
		\draw [style=Hadamard edge] (5) to (7);
		\draw [style=Hadamard edge] (6) to (4);
		\draw [style=Hadamard edge] (4) to (0);
		\draw [style=Hadamard edge] (5) to (0);
		\draw [style=Hadamard edge] (7) to (1);
		\draw [style=Hadamard edge] (6) to (1);
		\draw [style=Hadamard edge] (2) to (1);
		\draw [style=Hadamard edge] (2) to (0);
		\draw [style=Hadamard edge] (3) to (1);
		\draw [style=Hadamard edge] (0) to (1);
		\draw [style=Hadamard edge] (3) to (0);
		\draw [style=Hadamard edge] (2) to (4);
		\draw [style=Hadamard edge] (3) to (5);
		\draw [style=Hadamard edge] (3) to (6);
		\draw [style=Hadamard edge] (2) to (7);
		\draw (10.center) to (5);
		\draw (11.center) to (4);
		\draw (6) to (8.center);
		\draw (7) to (9.center);
		\draw (2) to (13.center);
		\draw (3) to (12.center);
		\draw [bend right=13] (0) to (15.center);
		\draw [bend left=13] (1) to (14.center);
	\end{pgfonlayer}
\end{tikzpicture} \begin{tikzpicture}[scale=0.4]
	\begin{pgfonlayer}{nodelayer}
		\node [style=Z dot] (0) at (2, 3) {};
		\node [style=Z dot] (1) at (2, -3) {};
		\node [style=Z dot] (2) at (0, -1) {};
		\node [style=Z dot] (3) at (0, 1) {};
		\node [style=Z dot] (4) at (-2.5, 3.25) {};
		\node [style=Z dot] (5) at (-4, 2) {};
		\node [style=Z dot] (6) at (-4, -2) {};
		\node [style=Z dot] (7) at (-2.5, -3.25) {};
		\node [style=none] (8) at (-4.5, -2.5) {};
		\node [style=none] (9) at (-3, -3.75) {};
		\node [style=none] (10) at (-4.5, 2.5) {};
		\node [style=none] (11) at (-3, 3.75) {};
		\node [style=none] (12) at (5, 1) {};
		\node [style=none] (13) at (5, -1) {};
		\node [style=none] (14) at (5, 3) {};
		\node [style=H gate] (16) at (3.5, 3) {};
		\node [style=H gate] (17) at (3.5, -3) {};
		\node [style=none] (18) at (2.5, 3.5) {$u$};
		\node [style=none] (19) at (2.5, -3.5) {$v$};
		\node [style=Z phase dot] (20) at (3, 1) {$\pi$};
		\node [style=Z phase dot] (21) at (3, -1) {$\pi$};
		\node [style=none] (22) at (5, -3) {};
	\end{pgfonlayer}
	\begin{pgfonlayer}{edgelayer}
		\draw [style=Hadamard edge] (4) to (0);
		\draw [style=Hadamard edge] (5) to (0);
		\draw [style=Hadamard edge] (7) to (1);
		\draw [style=Hadamard edge] (6) to (1);
		\draw [style=Hadamard edge] (2) to (1);
		\draw [style=Hadamard edge] (2) to (0);
		\draw [style=Hadamard edge] (3) to (1);
		\draw [style=Hadamard edge] (0) to (1);
		\draw [style=Hadamard edge] (3) to (0);
		\draw (10.center) to (5);
		\draw (11.center) to (4);
		\draw (6) to (8.center);
		\draw (7) to (9.center);
		\draw (1) to (17);
		\draw (0) to (16);
		\draw (16) to (14.center);
		\draw [style=Hadamard edge] (5) to (6);
		\draw [style=Hadamard edge] (4) to (7);
		\draw [style=Hadamard edge] (6) to (2);
		\draw [style=Hadamard edge] (7) to (3);
		\draw [style=Hadamard edge] (5) to (2);
		\draw [style=Hadamard edge] (4) to (3);
		\draw (3) to (20);
		\draw (20) to (12.center);
		\draw (2) to (21);
		\draw (21) to (13.center);
		\draw (17) to (22.center);
	\end{pgfonlayer}
\end{tikzpicture}
\end{center}
\end{lemma}

\begin{observation} \label{obs:inverses}
	Lemmas~\ref{lem:local-complementation} and \ref{lem:pivot} provide their own inverses: four successive local complementations about the same vertex or two successive pivots about the same edge leave a diagram invariant.
\end{observation}

\section{Converting planar measurements to XY-measurements}
In the \emph{graph-like diagrams} \cite{Duncan_2020} used in PyZX, all spiders are green and all edges are Hadamard edges. `Phase gadgets' consist of a degree-1 green spider connected to a phase-free green spider by a Hadamard edge as in the left-most diagram of \eqref{eq:graph-like-phase-gadget}.
When converting graph-like diagrams to MBQC-form, it is difficult to know whether to interpret phase gadgets as a single $\YZm$-measured vertex (middle diagram) or as an $X$-measured vertex connected to a degree-1 $\XYm$-measured vertex (right-most diagram):
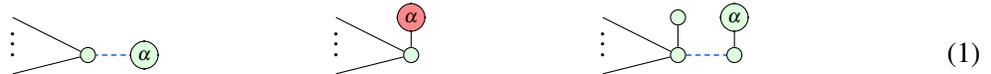
\begin{equation}\label{eq:graph-like-phase-gadget}
 \begin{tikzpicture}[baseline=-0.1cm]
	\begin{pgfonlayer}{nodelayer}
		\node [style=Z dot] (0) at (0, 0) {};
		\node [style=Z phase dot] (1) at (0.75, 0) {$\alpha$};
		\node [style=none] (2) at (-1, 0.5) {};
		\node [style=none] (3) at (-1, -0.25) {};
		\node [style=none] (4) at (-1, 0.25) {$\vdots$};
	\end{pgfonlayer}
	\begin{pgfonlayer}{edgelayer}
		\draw [style=hadamard edge] (1) to (0);
		\draw (0) to (2.center);
		\draw (0) to (3.center);
	\end{pgfonlayer}
\end{tikzpicture} \qquad\qquad\qquad \begin{tikzpicture}[baseline=-0.1cm]
	\begin{pgfonlayer}{nodelayer}
		\node [style=Z dot] (0) at (0, 0) {};
		\node [style=X phase dot] (1) at (0, 0.5) {$\alpha$};
		\node [style=none] (2) at (-1, 0.5) {};
		\node [style=none] (3) at (-1, -0.25) {};
		\node [style=none] (4) at (-1, 0.25) {$\vdots$};
	\end{pgfonlayer}
	\begin{pgfonlayer}{edgelayer}
		\draw (1) to (0);
		\draw (0) to (2.center);
		\draw (0) to (3.center);
	\end{pgfonlayer}
\end{tikzpicture} \qquad\qquad\qquad \begin{tikzpicture}[baseline=-0.1cm]
	\begin{pgfonlayer}{nodelayer}
		\node [style=Z dot] (0) at (0, 0) {};
		\node [style=Z phase dot] (1) at (0.75, 0.5) {$\alpha$};
		\node [style=none] (2) at (-1, 0.5) {};
		\node [style=none] (3) at (-1, -0.25) {};
		\node [style=none] (4) at (-1, 0.25) {$\vdots$};
		\node [style=Z dot] (5) at (0, 0.5) {};
		\node [style=Z dot] (6) at (0.75, 0) {};
	\end{pgfonlayer}
	\begin{pgfonlayer}{edgelayer}
		\draw (0) to (2.center);
		\draw (0) to (3.center);
		\draw [style=hadamard edge] (0) to (6);
		\draw (5) to (0);
		\draw (1) to (6);
	\end{pgfonlayer}
\end{tikzpicture}
\end{equation}
The following proposition shows that both interpretations are valid and can be interconverted.

\begin{proposition}\label{prop:YZtoXY}
	Let $(G,I,O,\ld)$ be a labelled open graph with Pauli flow where $G=(V,E)$, and suppose there exists $x\in V$ with $\ld(x) = \YZm$.
	Then $(G',I,O,\ld')$ has Pauli flow, where $V'=V\cup\{x'\}$, $E'=E\cup\{\{x,x'\}\}$, and $\ld'(x)=X$, $\ld'(x')=\XYm$, and $\ld'(v)=\ld(v)$ otherwise.
\end{proposition}
\begin{proof}
	Consider the following sequence of rewrites.
	
	
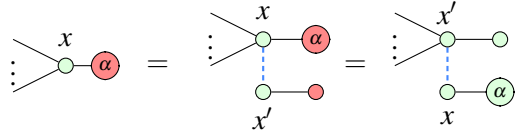
\begin{center}
		\begin{tikzpicture}[scale=0.7]
	\begin{pgfonlayer}{nodelayer}
		\node [style=Z dot] (0) at (-1.75, 0) {};
		\node [style=X phase dot] (1) at (-1, 0) {$\alpha$};
		\node [style=none] (2) at (-2.75, 0.5) {};
		\node [style=none] (3) at (-2.75, -0.5) {};
		\node [style=none] (4) at (-2.75, 0) {$\vdots$};
		\node [style=Z dot] (14) at (2, 0.5) {};
		\node [style=X phase dot] (15) at (3, 0.5) {$\alpha$};
		\node [style=none] (16) at (1, 1) {};
		\node [style=none] (17) at (1, 0) {};
		\node [style=none] (18) at (1, 0.5) {$\vdots$};
		\node [style=Z dot] (19) at (2, -0.5) {};
		\node [style=X dot] (20) at (3, -0.5) {};
		\node [style=Z dot] (21) at (5.5, 0.5) {};
		\node [style=none] (23) at (4.5, 1) {};
		\node [style=none] (24) at (4.5, 0) {};
		\node [style=none] (25) at (4.5, 0.5) {$\vdots$};
		\node [style=Z dot] (26) at (5.5, -0.5) {};
		\node [style=none] (29) at (0, 0) {$=$};
		\node [style=none] (30) at (3.75, 0) {$=$};
		\node [style=Z dot] (31) at (6.5, 0.5) {};
		\node [style=Z phase dot] (32) at (6.5, -0.5) {$\alpha$};
		\node [style=none] (33) at (-1.75, 0.5) {$x$};
		\node [style=none] (34) at (2, 1) {$x$};
		\node [style=none] (35) at (2, -1) {$x'$};
		\node [style=none] (36) at (5.5, 1) {$x'$};
		\node [style=none] (37) at (5.5, -1) {$x$};
	\end{pgfonlayer}
	\begin{pgfonlayer}{edgelayer}
		\draw (1) to (0);
		\draw (0) to (2.center);
		\draw (0) to (3.center);
		\draw (15) to (14);
		\draw (14) to (16.center);
		\draw (14) to (17.center);
		\draw [style=hadamard edge] (14) to (19);
		\draw (20) to (19);
		\draw (21) to (23.center);
		\draw (21) to (24.center);
		\draw [style=hadamard edge] (21) to (26);
		\draw (31) to (21);
		\draw (32) to (26);
	\end{pgfonlayer}
\end{tikzpicture}
	\end{center}
Here we insert the $Z$-measured vertex $x'$ connected only to $x$, then pivot along the edge $(x,x')$. Both $Z$-insertion and pivoting preserve the existence of Pauli flow, thus our new rewrite rule also preserves the existence of Pauli flow.
\end{proof}

A similar sequence of rewrites allows us to rewrite XZ-measurements in terms of just $Y$-measurements and XY-measurements

\begin{proposition}\label{prop:XZtoXY}
	Let $(G,I,O,\ld)$ be a labelled open graph with Pauli flow where $G=(V,E)$, and suppose there exists $x\in V$ with $\ld(x) = \XZm$.
	Then $(G',I,O,\ld')$ has Pauli flow, where $V'=V\cup\{x'\}$, $E'=E\cup\{\{x,x'\}\}$, and $\ld'(x)=Y$, $\ld'(x')=\XYm$ and $\ld'(v)=\ld(v)$ otherwise.
\end{proposition}

\begin{proof}
	Consider the following sequence of rewrites.
		
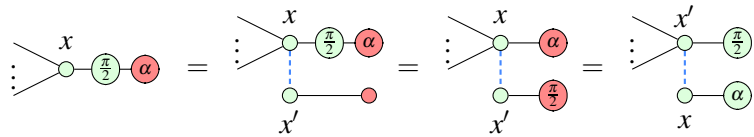
\begin{center}
		\begin{tikzpicture}[scale=0.7]
	\begin{pgfonlayer}{nodelayer}
		\node [style=Z dot] (0) at (-3.5, 0) {};
		\node [style=X phase dot] (1) at (-2, 0) {$\alpha$};
		\node [style=none] (2) at (-4.5, 0.5) {};
		\node [style=none] (3) at (-4.5, -0.5) {};
		\node [style=none] (4) at (-4.5, 0) {$\vdots$};
		\node [style=Z dot] (14) at (0.75, 0.5) {};
		\node [style=X phase dot] (15) at (2.25, 0.5) {$\alpha$};
		\node [style=none] (16) at (-0.25, 1) {};
		\node [style=none] (17) at (-0.25, 0) {};
		\node [style=none] (18) at (-0.25, 0.5) {$\vdots$};
		\node [style=Z dot] (19) at (0.75, -0.5) {};
		\node [style=X dot] (20) at (2.25, -0.5) {};
		\node [style=Z dot] (21) at (4.75, 0.5) {};
		\node [style=none] (23) at (3.75, 1) {};
		\node [style=none] (24) at (3.75, 0) {};
		\node [style=none] (25) at (3.75, 0.5) {$\vdots$};
		\node [style=Z dot] (26) at (4.75, -0.5) {};
		\node [style=none] (29) at (-1, 0) {$=$};
		\node [style=none] (30) at (3, 0) {$=$};
		\node [style=none] (33) at (-3.5, 0.5) {$x$};
		\node [style=none] (34) at (0.75, 1) {$x$};
		\node [style=none] (35) at (0.75, -1) {$x'$};
		\node [style=none] (36) at (4.75, 1) {$x$};
		\node [style=none] (37) at (4.75, -1) {$x'$};
		\node [style=Z phase dot] (38) at (-2.75, 0) {$\frac{\pi}{2}$};
		\node [style=Z phase dot] (39) at (1.5, 0.5) {$\frac{\pi}{2}$};
		\node [style=X phase dot] (40) at (5.75, 0.5) {$\alpha$};
		\node [style=X phase dot] (41) at (5.75, -0.5) {$\frac{\pi}{2}$};
		\node [style=Z dot] (42) at (8.25, 0.5) {};
		\node [style=none] (43) at (7.25, 1) {};
		\node [style=none] (44) at (7.25, 0) {};
		\node [style=none] (45) at (7.25, 0.5) {$\vdots$};
		\node [style=Z dot] (46) at (8.25, -0.5) {};
		\node [style=none] (47) at (8.25, 1) {$x'$};
		\node [style=none] (48) at (8.25, -1) {$x$};
		\node [style=Z phase dot] (51) at (9.25, 0.5) {$\frac{\pi}{2}$};
		\node [style=Z phase dot] (52) at (9.25, -0.5) {$\alpha$};
		\node [style=none] (53) at (6.5, 0) {$=$};
	\end{pgfonlayer}
	\begin{pgfonlayer}{edgelayer}
		\draw (0) to (2.center);
		\draw (0) to (3.center);
		\draw (14) to (16.center);
		\draw (14) to (17.center);
		\draw [style=hadamard edge] (14) to (19);
		\draw (20) to (19);
		\draw (21) to (23.center);
		\draw (21) to (24.center);
		\draw [style=hadamard edge] (21) to (26);
		\draw (0) to (38);
		\draw (38) to (1);
		\draw (14) to (39);
		\draw (39) to (15);
		\draw (21) to (40);
		\draw (26) to (41);
		\draw (42) to (43.center);
		\draw (42) to (44.center);
		\draw [style=hadamard edge] (42) to (46);
		\draw (42) to (51);
		\draw (46) to (52);
	\end{pgfonlayer}
\end{tikzpicture}
	\end{center}
	
	Here we insert a $Z$-measured vertex $x'$ connected to the $XZ$-measured vertex $x$, perform local complementation about $x'$, then pivot along the edge $(x,x')$. Each of these rewrites preserves the existence of Pauli flow, thus the resulting pattern has Pauli flow.
\end{proof}

Using the two previous propositions, we are able to re-write any XZ- and YZ-planar measurements into a Pauli measurement plus an XY-measurement. This implies the following.

\begin{proposition}\label{prop:all-XY}
	Let $(G,I,O,\lambda)$ be an arbitrary MBQC-form diagram with Pauli flow. Then there exists an equivalent diagram $(G', I', O', \lambda')$ with Pauli flow where $\lambda'(v) \in \{X, Y, XY\}$ for all $v \in V'$. 
\end{proposition}

\begin{proof}
	We begin by applying $Z$-deletion (Lemma~\ref{lem:Z-delete}) to all $Z$-measured vertices, leaving us with only $X$, $Y$, $XY$, $XZ$ and $YZ$ vertices. It remains to remove all $XZ$ and $YZ$ measurements.
	
	By Proposition~\ref{prop:YZtoXY}, we can convert every $YZ$-measured vertex into an $X$-measured vertex connected to an $XY$-measured vertex while preserving the existence of Pauli flow. Then, by Proposition~\ref{prop:XZtoXY} we can convert every $XZ$-measured vertex into a $Y$-measured vertex connected to an $XY$-measured vertex. We now only have $X$, $Y$ and $XY$ measured vertices remaining, and each rewrite rule used to get here preserves the existence of Pauli flow, thus the resulting graph has Pauli flow.
\end{proof}

\begin{remark}
	Note that Pauli flow is important here: 
the gflow conditions need not be satisfied if the newly-introduced  Pauli measurements were taken to be arbitrary XY-measurements instead.

For example, the first of the following two diagrams has a gflow $(g,\prec)$ with $g(a)=\{c\}$, $g(b)=\{d\}$, $g(x)=\{c,d,x\}$ and $a,b,x\prec c,d$.
The second diagram has Pauli flow by Proposition~\ref{prop:YZtoXY}, but it does not have gflow: any flow $(p,\prec')$ on the second diagram must have $x\in p(x')$ to satisfy $x'\in\odd{p(x')}$, as inputs $a,b$ do not appear in correction sets. Similarly, $x'\in p(x)$ as it is the only neighbour. Thus the gflow conditions would require $x\prec' x'$ and $x'\prec' x$ simultaneously, which is not possible.
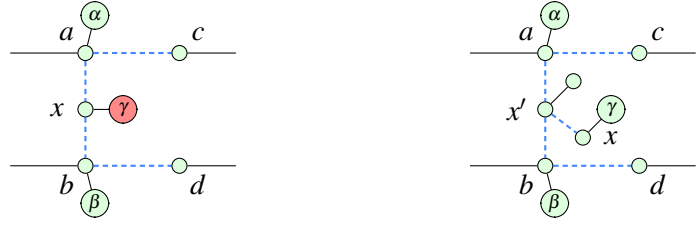
\[
 \begin{tikzpicture}[scale=0.5, baseline=-0.1cm]
	\begin{pgfonlayer}{nodelayer}
		\node [style=none] (0) at (-4, 1.5) {};
		\node [style=none] (1) at (-4, -1.5) {};
		\node [style=Z dot] (2) at (-2, 1.5) {};
		\node [style=Z dot] (3) at (-2, -1.5) {};
		\node [style=Z dot] (4) at (0.5, 1.5) {};
		\node [style=Z dot] (5) at (0.5, -1.5) {};
		\node [style=none] (6) at (2, 1.5) {};
		\node [style=none] (7) at (2, -1.5) {};
		\node [style=Z phase dot] (8) at (-1.75, 2.5) {$\alpha$};
		\node [style=Z phase dot] (9) at (-1.75, -2.5) {$\beta$};
		\node [style=none] (10) at (-2.5, 2) {$a$};
		\node [style=none] (11) at (-2.5, -2) {$b$};
		\node [style=none] (12) at (1, 2) {$c$};
		\node [style=none] (13) at (1, -2) {$d$};
		\node [style=Z dot] (14) at (-2, 0) {};
		\node [style=X phase dot] (15) at (-1, 0) {$\gamma$};
		\node [style=none] (16) at (-2.75, 0) {$x$};
	\end{pgfonlayer}
	\begin{pgfonlayer}{edgelayer}
		\draw [style=hadamard edge] (3) to (5);
		\draw [style=hadamard edge] (2) to (4);
		\draw (0.center) to (2);
		\draw (4) to (6.center);
		\draw (1.center) to (3);
		\draw (5) to (7.center);
		\draw (8) to (2);
		\draw (3) to (9);
		\draw [style=hadamard edge] (2) to (14);
		\draw [style=hadamard edge] (14) to (3);
		\draw (14) to (15);
	\end{pgfonlayer}
\end{tikzpicture} \qquad\qquad\qquad\qquad \begin{tikzpicture}[scale=0.5, baseline=-0.1cm]
	\begin{pgfonlayer}{nodelayer}
		\node [style=none] (0) at (-4, 1.5) {};
		\node [style=none] (1) at (-4, -1.5) {};
		\node [style=Z dot] (2) at (-2, 1.5) {};
		\node [style=Z dot] (3) at (-2, -1.5) {};
		\node [style=Z dot] (4) at (0.5, 1.5) {};
		\node [style=Z dot] (5) at (0.5, -1.5) {};
		\node [style=none] (6) at (2, 1.5) {};
		\node [style=none] (7) at (2, -1.5) {};
		\node [style=Z phase dot] (8) at (-1.75, 2.5) {$\alpha$};
		\node [style=Z phase dot] (9) at (-1.75, -2.5) {$\beta$};
		\node [style=none] (10) at (-2.5, 2) {$a$};
		\node [style=none] (11) at (-2.5, -2) {$b$};
		\node [style=none] (12) at (1, 2) {$c$};
		\node [style=none] (13) at (1, -2) {$d$};
		\node [style=Z dot] (14) at (-2, 0) {};
		\node [style=Z phase dot] (15) at (-0.25, 0) {$\gamma$};
		\node [style=none] (16) at (-2.75, 0) {$x'$};
		\node [style=Z dot] (17) at (-1.25, 0.75) {};
		\node [style=Z dot] (18) at (-1, -0.75) {};
		\node [style=none] (19) at (-0.25, -0.75) {$x$};
	\end{pgfonlayer}
	\begin{pgfonlayer}{edgelayer}
		\draw [style=hadamard edge] (3) to (5);
		\draw [style=hadamard edge] (2) to (4);
		\draw (0.center) to (2);
		\draw (4) to (6.center);
		\draw (1.center) to (3);
		\draw (5) to (7.center);
		\draw (8) to (2);
		\draw (3) to (9);
		\draw [style=hadamard edge] (2) to (14);
		\draw [style=hadamard edge] (14) to (3);
		\draw (17) to (14);
		\draw (15) to (18);
		\draw [style=hadamard edge] (14) to (18);
	\end{pgfonlayer}
\end{tikzpicture}
\]
\end{remark}

\section{Subdividing an edge}
 	Research on flow-preserving rewrite rules so far has been geared towards optimization, which usually involves reducing the number of vertices in a pattern. Yet there are also cases where it is desirable to introduce new vertices. An example of this is the obfuscation protocol for blind quantum computing of~\cite{Cao_2022}, which used an unpublished rewrite rule proved by one of the authors. We give the proof below.
	
\begin{proposition}
	Let $G=(V,E)$ be a graph with vertices $V$ and edges $E$.
	Suppose the labelled open graph $(G,I,O,\ld)$ has Pauli flow.
	Pick an edge $\{v,w\}\in E$ and subdivide it twice, i.e.\ let $G':=(V',E')$, where $V':=V\cup\{v',w'\}$ contains two new vertices $v',w'$, and
	\[
	E' := (E\setminus\{\{v,w\}\})\cup\{\{v,w'\},\{w',v'\},\{v',w\}\}.
	\]
	Then $(G',I,O,\ld')$ has Pauli flow, where $\ld'(v')=\ld'(w')=X$ and $\ld'(u)=\ld(u)$ for all $u\in V\setminus O$.
\end{proposition}

\begin{proof}
	We may subdivide an edge by inserting two $Z$-measured vertices as shown in the following diagram, then pivoting about these two $Z$-measured vertices.
	
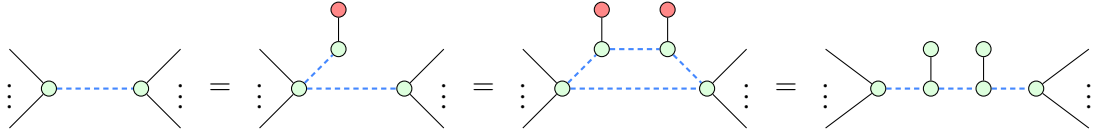
\begin{center}
		\begin{tikzpicture}[scale=0.7]
	\begin{pgfonlayer}{nodelayer}
		\node [style=Z dot] (0) at (-4, 0) {};
		\node [style=Z dot] (1) at (-2.25, 0) {};
		\node [style=none] (4) at (-4.75, 0.75) {};
		\node [style=none] (5) at (-4.75, -0.75) {};
		\node [style=none] (6) at (-1.5, -0.75) {};
		\node [style=none] (7) at (-1.5, 0.75) {};
		\node [style=Z dot] (8) at (5.75, 0) {};
		\node [style=Z dot] (9) at (8.5, 0) {};
		\node [style=none] (12) at (5, 0.75) {};
		\node [style=none] (13) at (5, -0.75) {};
		\node [style=none] (14) at (9.25, -0.75) {};
		\node [style=none] (15) at (9.25, 0.75) {};
		\node [style=Z dot] (16) at (6.5, 0.75) {};
		\node [style=Z dot] (17) at (7.75, 0.75) {};
		\node [style=X dot] (18) at (6.5, 1.5) {};
		\node [style=X dot] (19) at (7.75, 1.5) {};
		\node [style=none] (20) at (-0.75, 0) {$=$};
		\node [style=Z dot] (21) at (11.75, 0) {};
		\node [style=Z dot] (22) at (14.75, 0) {};
		\node [style=none] (25) at (10.75, 0.75) {};
		\node [style=none] (26) at (10.75, -0.75) {};
		\node [style=none] (27) at (15.75, -0.75) {};
		\node [style=none] (28) at (15.75, 0.75) {};
		\node [style=Z dot] (29) at (12.75, 0) {};
		\node [style=Z dot] (30) at (13.75, 0) {};
		\node [style=none] (33) at (4.25, 0) {$=$};
		\node [style=Z dot] (34) at (12.75, 0.75) {};
		\node [style=Z dot] (35) at (13.75, 0.75) {};
		\node [style=Z dot] (36) at (0.75, 0) {};
		\node [style=Z dot] (37) at (2.75, 0) {};
		\node [style=none] (38) at (0, 0.75) {};
		\node [style=none] (39) at (0, -0.75) {};
		\node [style=none] (40) at (3.5, -0.75) {};
		\node [style=none] (41) at (3.5, 0.75) {};
		\node [style=Z dot] (42) at (1.5, 0.75) {};
		\node [style=X dot] (44) at (1.5, 1.5) {};
		\node [style=none] (45) at (10, 0) {$=$};
		\node [style=none] (46) at (-4.75, 0) {$\vdots$};
		\node [style=none] (47) at (-1.5, 0) {$\vdots$};
		\node [style=none] (48) at (0, 0) {$\vdots$};
		\node [style=none] (49) at (3.5, 0) {$\vdots$};
		\node [style=none] (50) at (5, 0) {$\vdots$};
		\node [style=none] (51) at (9.25, 0) {$\vdots$};
		\node [style=none] (52) at (10.75, 0) {$\vdots$};
		\node [style=none] (53) at (15.75, 0) {$\vdots$};
	\end{pgfonlayer}
	\begin{pgfonlayer}{edgelayer}
		\draw (0) to (4.center);
		\draw (0) to (5.center);
		\draw (6.center) to (1);
		\draw (1) to (7.center);
		\draw [style=hadamard edge] (0) to (1);
		\draw (8) to (12.center);
		\draw (8) to (13.center);
		\draw (14.center) to (9);
		\draw (9) to (15.center);
		\draw [style=hadamard edge] (8) to (9);
		\draw [style=hadamard edge] (8) to (16);
		\draw [style=hadamard edge] (16) to (17);
		\draw [style=hadamard edge] (17) to (9);
		\draw (18) to (16);
		\draw (19) to (17);
		\draw (21) to (25.center);
		\draw (21) to (26.center);
		\draw (27.center) to (22);
		\draw (22) to (28.center);
		\draw [style=hadamard edge] (21) to (29);
		\draw [style=hadamard edge] (29) to (30);
		\draw [style=hadamard edge] (30) to (22);
		\draw (34) to (29);
		\draw (35) to (30);
		\draw (36) to (38.center);
		\draw (36) to (39.center);
		\draw (40.center) to (37);
		\draw (37) to (41.center);
		\draw [style=hadamard edge] (36) to (37);
		\draw [style=hadamard edge] (36) to (42);
		\draw (44) to (42);
	\end{pgfonlayer}
\end{tikzpicture}
	\end{center}
As inserting $Z$-measured vertices and pivoting both preserve the existence of Pauli flow, subdividing an edge also preserves the existence of Pauli flow.
\end{proof}

\section{Splitting a vertex}
Each of the previously mentioned Pauli-flow preserving rewrite rules only changes measurement angles by integer multiples of $\frac{\pi}{2}$. Here we introduce the first Pauli-flow preserving rewrite rule which allows us to change measurement angles arbitrarily.
To simplify the proof, the proposition requires that all measurements in the pattern are XY, $X$ or $Y$; by Proposition~\ref{prop:all-XY} this is without loss of generality.

\begin{proposition}\label{splittingproof}
	Let $G=(V,E)$ be a graph with vertices $V$ and edges $E$.
	Suppose the labelled open graph $(G,I,O,\ld)$ has Pauli flow and satisfies  $\ld(u)\in\{\XYm,X, Y\}$ for all $u\in O^c$.
	Pick a vertex $a\in O^c$ such that $\ld(a)=\XYm$ and split it, i.e.\ let $G':=(V',E')$, where $V':=V\cup\{x,a'\}$ contains two new vertices $x,a'$, and choose some (possibly empty) subset $W\sse N(x)$ such that
	\[
	E' := (E\setminus\{\{a,w\}\mid w\in W\})\cup\{\{a',w\}\mid w\in W\}\cup\{\{a,x\},\{x,a'\}\}.
	\]
	Then $(G',I,O,\ld')$ has Pauli flow, where $\ld'(x)=X$, $\ld'(a')=\XYm$, and $\ld'(u)=\ld(u)$ for all $u\in V\setminus O$.
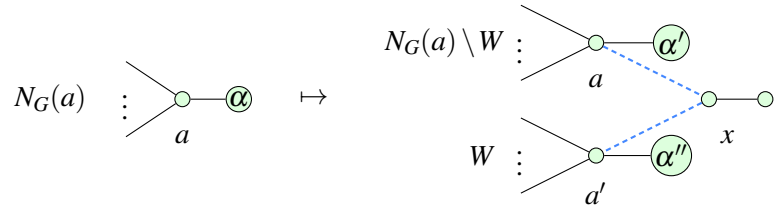
	\begin{center}
	\begin{tikzpicture}
	\begin{pgfonlayer}{nodelayer}
		\node [style=Z dot] (0) at (-2.25, 0) {};
		\node [style=none] (2) at (-3, 0.5) {};
		\node [style=none] (3) at (-3, -0.5) {};
		\node [style=none] (6) at (-3, 0) {$\vdots$};
		\node [style=Z dot] (17) at (4.75, 0) {};
		\node [style=Z dot] (22) at (5.5, 0) {};
		\node [style=Z dot] (23) at (3.25, 0.75) {};
		\node [style=Z dot] (24) at (3.25, -0.75) {};
		\node [style=none] (25) at (2.25, 1.25) {};
		\node [style=none] (26) at (2.25, 0.25) {};
		\node [style=none] (27) at (2.25, -1.25) {};
		\node [style=none] (28) at (2.25, -0.25) {};
		\node [style=none] (29) at (2.25, 0.75) {$\vdots$};
		\node [style=none] (33) at (-0.5, 0) {$\mapsto$};
		\node [style=none] (34) at (-2.25, -0.5) {$a$};
		\node [style=none] (36) at (3.25, 0.25) {$a$};
		\node [style=none] (37) at (5, -0.5) {$x$};
		\node [style=none] (38) at (3.25, -1.25) {$a'$};
		\node [style=Z dot] (39) at (-1.5, 0) {$\alpha$};
		\node [style=Z dot] (40) at (4.25, -0.75) {$\alpha''$};
		\node [style=Z dot] (41) at (4.25, 0.75) {$\alpha'$};
		\node [style=none] (42) at (2.25, -0.75) {$\vdots$};
		\node [style=none] (43) at (-4, 0) {$N_G(a)$};
		\node [style=none] (44) at (1.75, -0.75) {$W$};
		\node [style=none] (45) at (1.25, 0.75) {$N_G(a)\setminus W$};
	\end{pgfonlayer}
	\begin{pgfonlayer}{edgelayer}
		\draw (2.center) to (0);
		\draw (0) to (3.center);
		\draw (22) to (17);
		\draw (25.center) to (23);
		\draw (23) to (26.center);
		\draw (27.center) to (24);
		\draw (24) to (28.center);
		\draw (39) to (0);
		\draw (41) to (23);
		\draw (40) to (24);
		\draw [style=hadamard edge] (23) to (17);
		\draw [style=hadamard edge] (17) to (24);
	\end{pgfonlayer}
\end{tikzpicture}
	\end{center}
\end{proposition}

\begin{proof}
	Let $(p,\prec)$ be a focused Pauli flow for $(G,I,O,\ld)$; this exists as the labelled open graph has Pauli flow.
	Since all measurements are $XY$, $X$ or $Y$, the focusing conditions from Definition~\ref{def:focused_Pauli-flow} reduce to:
	\begin{itemize}
		\item For all $u \in V\setminus O$, if $v \in  \odds{G}{p(u)}\setminus (O \cup \{u\})$ then $\lambda(v) = Y$.
		\item For all $u \in V\setminus O$ and all $v \in V\setminus (O\cup\{u\})$ such that $\lambda(v) = Y$, we have $v \in p(u) \leftrightarrow v \in \odds{G}{p(u)}$.
	\end{itemize}
	Now, for all $u\in V\setminus O$, define
	\[
	p'(u) := \begin{cases}
		p(u)\cup\{x,a'\} &\text{if } a\in p(u) \text{ and } \abs{p(u)\cap W}\equiv 1\pmod 2 \\
		p(u)\cup\{a'\} &\text{if } a\in p(u) \text{ and } \abs{p(u)\cap W}\equiv 0\pmod 2 \\
		p(u)\cup\{x\} &\text{if } a\notin p(u) \text{ and } \abs{p(u)\cap W}\equiv 1\pmod 2 \\
		p(u) &\text{if } a\notin p(u) \text{ and } \abs{p(u)\cap W}\equiv 0\pmod 2,
	\end{cases}
	\]
	then it is straightforward to check that $\odds{G'}{p'(u)}= \odds{G}{p(u)}$.
	For example, in the first case, note that $\cup$ can be replaced by $\symd$ since $x,a'$ are new vertices that cannot appear in $p(u)$.
	Thus
	\begin{align*}
		\odds{G'}{p'(u)} &= \odds{G'}{p(u)\symd\{x,a'\}} \\
		&= \odds{G'}{p(u)} \symd \odds{G'}{\{x,a'\}} \\
		&= \odds{G'}{a} \symd \left(\Symdi{w\in p(u)\setminus (W\cup\{a\}\cup O)} \odds{G'}{w} \right) \symd \left(\Symdi{w\in (p(u)\cap W)\setminus O} \odds{G'}{w} \right) \symd \{a,x,a'\} \symd W \\
		&= \odds{G}{a}\symd \left(\Symdi{w\in p(u)\setminus (W\cup\{a\}\cup O)} \odds{G}{w} \right) \symd \left(\Symdi{w\in (p(u)\cap W)\setminus O} \odds{G}{w}\symd\{a,a'\} \right) \symd \{a,a'\} \\
		&= \odds{G}{p(u)},
	\end{align*}
	where the third step uses $\odds{G'}{a} = \odds{G}{a}\symd W \symd\{x\}$, and the final step uses $\abs{p(u)\cap W}\equiv 1\pmod 2$.
	
	If $a$ is an input in $G$, then it remains an input in $G'$ (the `input' is not a neighbour so cannot be transferred to $a'$ during the splitting process).
	This is without loss of generality: if $a'$ is desired to be an input, replace $W$ by $N_G(a)\setminus W$ and swap the labels $\alpha'$ and $\alpha''$ to get a \LOG{} that is equivalent to the desired one up to relabelling $a \leftrightarrow a'$.
	Having $a$ be an input is compatible with the correction set for $x$ in the next step below.
	If $a$ is not an input, there is actually a choice of whether to correct $x$ via $a$ or $a'$; we shall always choose the latter for some slight notational convenience.
	
	Let $p'(x) := \{a'\}$, and let $p'(a') := p'(a)\symd\{x\}$,
	resulting in the following odd neighbourhoods:
	\begin{align} 
	\odds{G'}{p'(x)} &= \odds{G'}{\{a'\}} = W \cup \{x\} \label{eq:1}\\
	\intertext{and}
	\odds{G'}{p'(a')}
	&= \odds{G'}{p'(a)} \symd \odds{G'}{\{x\}} \nonumber \\
	&= \odds{G}{p(a)} \symd \{a, a'\} \nonumber \\
	&= (\odds{G}{p(a)} \cup \{a'\})\setminus \{a\} \label{eq:2}
	\end{align}
	where the final step uses the fact that $a\in \odds{G}{p(a)}$ and $a' \not \in \odds{G}{p(a)}$ (since $a'$ is not even in $G$).
	
	Let $\prec'$ be the transitive closure of
	\[
		{\prec} \cup \Big\{(w,a')\mid w\prec a\Big\} \cup \Big\{(a',w)\mid a\prec w\Big\} \cup \Big\{(x,w)\mid w\in W\Big\} \cup \Big\{ (x,a')\Big\}.
	\]
	This is a partial order since $a'$ has the same relationships as $a$ (except for being a successor of $x$) and $x$ only has successors.
	
	The proof that $(p', \prec')$ is a Pauli flow for $G'$ can be found in Appendix~\ref{appendix}.
\end{proof}

We are able to obtain other useful rewrite rules as immediate corollaries of this.
\begin{corollary}
	Using Proposition~\ref{splittingproof} with $W = \emptyset$ and $\alpha'' = 0$, we obtain the following rule used in \cite{Cao_2022}.
	
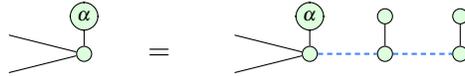
	\begin{center}
		\begin{tikzpicture}
	\begin{pgfonlayer}{nodelayer}
		\node [style=Z dot] (0) at (-1, 1) {};
		\node [style=none] (1) at (-2, 1.25) {};
		\node [style=none] (2) at (-2, 0.75) {};
		\node [style=Z phase dot] (3) at (-1, 1.5) {$\alpha$};
		\node [style=none] (4) at (0, 1) {$=$};
		\node [style=Z dot] (9) at (2, 1) {};
		\node [style=none] (10) at (1, 1.25) {};
		\node [style=none] (11) at (1, 0.75) {};
		\node [style=Z phase dot] (12) at (2, 1.5) {$\alpha$};
		\node [style=Z dot] (13) at (3, 1) {};
		\node [style=Z dot] (14) at (4, 1) {};
		\node [style=Z dot] (15) at (3, 1.5) {};
		\node [style=Z dot] (16) at (4, 1.5) {};
	\end{pgfonlayer}
	\begin{pgfonlayer}{edgelayer}
		\draw (0) to (1.center);
		\draw (0) to (2.center);
		\draw (3) to (0);
		\draw (9) to (10.center);
		\draw (9) to (11.center);
		\draw (12) to (9);
		\draw [style=hadamard edge] (9) to (13);
		\draw [style=hadamard edge] (13) to (14);
		\draw (15) to (13);
		\draw (16) to (14);
	\end{pgfonlayer}
\end{tikzpicture}
	\end{center}
\end{corollary}

This rule can alternatively be derived in a more round-about way from $Z$-insertion and pivoting, but we next prove a rule that truly requires vertex splitting.

\section{Neighbour unfusion}
\label{s:neighbour-unfusion}
In \cite{korbinian_2022}, a rewrite rule called \textit{neighbour unfusion} was used to reduce the number of two-qubit gates in circuits via the ZX-calculus. Using neighbour unfusion allowed for the two-qubit gate count to be greatly reduced, but introduced a new problem: neighbour unfusion, which introduces two new qubits in each application, was found to not always preserve gflow.
Yet a flow is needed to be able to translate back to a circuit after the application of the two-qubit gate count reduction algorithm.
We now show that neighbour unfusion preserves the existence of Pauli flow, so circuit re-extraction is always possible.

\begin{corollary}
	By applying vertex splitting with $|W| = 1$, we obtain the following `neighbour unfusion' rule, where $\alpha = \alpha' + \alpha''$ (the measurement for the right-most vertex is not drawn as it can be measured in any plane, or even be an output).
	
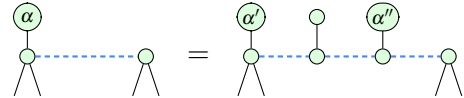
\begin{center}
		\begin{tikzpicture}[scale=0.7]
	\begin{pgfonlayer}{nodelayer}
		\node [style=Z dot] (0) at (-4, 0) {};
		\node [style=Z dot] (1) at (-1.75, 0) {};
		\node [style=none] (4) at (-4.25, -0.75) {};
		\node [style=none] (5) at (-3.75, -0.75) {};
		\node [style=none] (6) at (-2, -0.75) {};
		\node [style=none] (7) at (-1.5, -0.75) {};
		\node [style=none] (20) at (-0.75, 0) {$=$};
		\node [style=Z dot] (21) at (0.25, 0) {};
		\node [style=Z dot] (22) at (4, 0) {};
		\node [style=none] (25) at (0, -0.75) {};
		\node [style=none] (26) at (0.5, -0.75) {};
		\node [style=none] (27) at (3.75, -0.75) {};
		\node [style=none] (28) at (4.25, -0.75) {};
		\node [style=Z dot] (29) at (1.5, 0) {};
		\node [style=Z dot] (30) at (2.75, 0) {};
		\node [style=Z dot] (34) at (1.5, 0.75) {};
		\node [style=Z phase dot] (35) at (2.75, 0.75) {$\alpha''$};
		\node [style=Z phase dot] (45) at (-4, 0.75) {$\alpha$};
		\node [style=Z phase dot] (48) at (0.25, 0.75) {$\alpha'$};
	\end{pgfonlayer}
	\begin{pgfonlayer}{edgelayer}
		\draw (0) to (4.center);
		\draw (0) to (5.center);
		\draw (6.center) to (1);
		\draw (1) to (7.center);
		\draw [style=hadamard edge] (0) to (1);
		\draw (21) to (25.center);
		\draw (21) to (26.center);
		\draw (27.center) to (22);
		\draw (22) to (28.center);
		\draw [style=hadamard edge] (21) to (29);
		\draw [style=hadamard edge] (29) to (30);
		\draw [style=hadamard edge] (30) to (22);
		\draw (34) to (29);
		\draw (35) to (30);
		\draw (45) to (0);
		\draw (48) to (21);
	\end{pgfonlayer}
\end{tikzpicture}
	\end{center}
\end{corollary}

Staudacher et al.\ \cite{korbinian_2022} state that, in the case of only XY-measurements, neighbour unfusion empirically fails to preserve gflow if the two vertices to which neighbour unfusion is applied are extracted to different qubits in the circuit extraction process.
While we have not fully formalised this idea, we give a condition which guarantees that neighbour unfusion preserves gflow.

\begin{proposition}\label{prop:neighbour-unfusion-gflow-sufficient}
Let $\Gamma = (G, I, O, \lambda)$ be a labelled open graph and  suppose $a,b$ are two adjacent vertices in $G$ with $\ld(a) = \ld(b) = \XYm$.
Suppose $\Gamma$ has focused gflow $(g,\prec)$ where $b\in g(a)$ and for all $w \in V \setminus \{a, b\}$, $w \prec b \implies w \prec a$ and $a \prec w\implies b \prec w$.
Let $\Gamma'$ be the labelled open graph after applying neighbour unfusion to $a$ and $b$.
Then $\Gamma'$ has gflow.
The same holds with the roles of $a$ and $b$ reversed.
\end{proposition}

\begin{proof}
	Consider a labelled open graph which has a focused gflow $(g,\prec)$ and contains the subdiagram
	
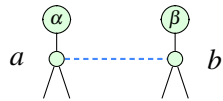
\begin{center}
		\begin{tikzpicture}[scale=0.7]
	\begin{pgfonlayer}{nodelayer}
		\node [style=Z dot] (0) at (-1, 0) {};
		\node [style=Z dot] (1) at (1.25, 0) {};
		\node [style=none] (4) at (-1.25, -0.75) {};
		\node [style=none] (5) at (-0.75, -0.75) {};
		\node [style=none] (6) at (1, -0.75) {};
		\node [style=none] (7) at (1.5, -0.75) {};
		\node [style=Z phase dot] (45) at (-1, 0.75) {$\alpha$};
		\node [style=none] (46) at (-1.75, 0) {$a$};
		\node [style=none] (47) at (2, 0) {$b$};
		\node [style=Z phase dot] (48) at (1.25, 0.75) {$\beta$};
	\end{pgfonlayer}
	\begin{pgfonlayer}{edgelayer}
		\draw (0) to (4.center);
		\draw (0) to (5.center);
		\draw (6.center) to (1);
		\draw (1) to (7.center);
		\draw [style=hadamard edge] (0) to (1);
		\draw (45) to (0);
		\draw (48) to (1);
	\end{pgfonlayer}
\end{tikzpicture}
	\end{center}
Assume without loss of generality that $b \in g(a)$ and that for all $w \in V \setminus \{a, b\}$, $w \prec b$ implies $w \prec a$ and $a \prec w$ implies $b \prec w$; the other case is symmetric. Neighbour unfusion yields a labelled open graph $\Gamma'=(G',I,O,\ld')$ with the following subdiagram, where $\alpha'+\alpha''=\alpha$:
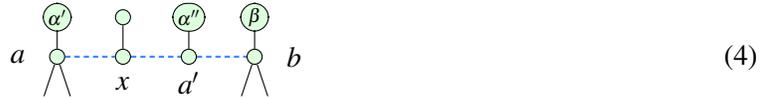
\begin{equation}\label{eq:after_neighbour_unfusion}
	\begin{tikzpicture}[scale=0.7, baseline=-0.1cm]
	\begin{pgfonlayer}{nodelayer}
		\node [style=Z dot] (21) at (-1.75, 0) {};
		\node [style=Z dot] (22) at (2, 0) {};
		\node [style=none] (25) at (-2, -0.75) {};
		\node [style=none] (26) at (-1.5, -0.75) {};
		\node [style=none] (27) at (1.75, -0.75) {};
		\node [style=none] (28) at (2.25, -0.75) {};
		\node [style=Z dot] (29) at (-0.5, 0) {};
		\node [style=Z dot] (30) at (0.75, 0) {};
		\node [style=Z dot] (34) at (-0.5, 0.75) {};
		\node [style=Z phase dot] (35) at (0.75, 0.75) {$\alpha''$};
		\node [style=Z phase dot] (48) at (-1.75, 0.75) {$\alpha'$};
		\node [style=Z phase dot] (49) at (2, 0.75) {$\beta$};
		\node [style=none] (50) at (-2.5, 0) {$a$};
		\node [style=none] (51) at (2.75, 0) {$b$};
		\node [style=none] (52) at (-0.5, -0.5) {$x$};
		\node [style=none] (53) at (0.75, -0.5) {$a'$};
	\end{pgfonlayer}
	\begin{pgfonlayer}{edgelayer}
		\draw (21) to (25.center);
		\draw (21) to (26.center);
		\draw (27.center) to (22);
		\draw (22) to (28.center);
		\draw [style=hadamard edge] (21) to (29);
		\draw [style=hadamard edge] (29) to (30);
		\draw [style=hadamard edge] (30) to (22);
		\draw (34) to (29);
		\draw (35) to (30);
		\draw (48) to (21);
		\draw (49) to (22);
	\end{pgfonlayer}
\end{tikzpicture}
\end{equation}
We can construct a focused gflow for the new pattern by defining the correction sets as follows:
\[
	g'(v) = \begin{cases} g(v) \cup \{x, a'\} &\text{if } a \in g(v) \wedge b \in g(v) \\ g(v) \cup \{a'\} &\text{if } a \in g(v) \wedge b \not \in g(v) \\ g(v)\cup \{x\} &\text{if } a \not \in g(v) \wedge b \in g(v)\\ g(b) \cup \{a'\} &\text{if } v=x\\ g(a)	 &\text{if } v= a' \\ g(v) &\text{otherwise.}\end{cases}
\]
This choice leaves invariant the odd neighbourhoods of the correction sets of any original (non-output) vertices.
Furthermore we have
\[
 \odds{G'}{g'(x)}
 = \odds{G'}{g(b)\cup\{a'\}}
 = \odds{G'}{g(b)\symd\{a'\}}
 = \odds{G}{g(b)}\symd\{x,b\}
\]
since $a,a'\notin g(b)$, and
\[
 \odds{G'}{g'(a')}
 = \odds{G'}{g(a)}
 = \odds{G}{g(a)}\symd\{a,a'\}
\]
since $b\in g(a)$.
Therefore $x\in\odds{G'}{g'(x)}$ and $a'\in\odds{G'}{g'(a')}$ as desired, and furthermore the correction sets for the new vertices are focused.
Take $\prec'$ to be the transitive closure of ${\prec} \cup \{(a,x), (x, a'), (a', b)\}$, then $(g', \prec')$ is a focused gflow for $\Gamma'$: Firstly, the relation $\prec'$ is a strict partial order.
To show that the gflow conditions are satisfied by $(g',\prec')$, it suffices to consider that the modified correction sets are compatible with the new partial order:
\begin{itemize}
 \item For any $v\in V$, we have $a'\in g'(v)$ only if $a\in g(v)$.
  The latter implies $v\prec a$ and thus $v\prec' a\prec' a'$.
 \item For any $v\in V$, we have $x\in g'(v)$ only if $b\in g(v)$.
  The latter implies $v\prec b$ and thus by assumption $v\prec a$.
  Then, as in the previous case, $v\prec' a\prec' x$.
 \item If $w\in g'(x)$, then either $w=a'$ or $w\in g(b)$.
  The former is straightforward as $x\prec' a'$ by definition.
  For the latter, we have $b\prec w$ since $(g,\prec)$ is a gflow, and thus $x\prec' b\prec' w$.
 \item If $w\in g'(a')$, then $w\in g(a)$.
  This implies $a\prec w$ since $(g,\prec)$ is a gflow, and furthermore $b\prec w$ by assumption.
  Thus, $a'\prec' b\prec' w$.
\end{itemize}
All of the other gflow conditions are satisfied as $(g, \prec)$ is a gflow for $\Gamma$.
\end{proof}

We will illustrate the neighbour unfusion process with an example that shows some choices of unfusion which do preserve gflow and others which do not.

\begin{example}
	Consider the following MBQC-form diagram, which appeared in a different context in~\cite{miyazaki_analysis_2015}.
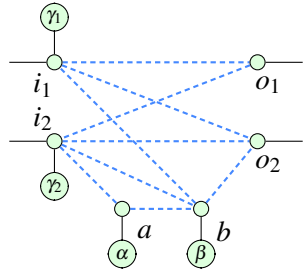
	\begin{center}
		\begin{tikzpicture}[scale=0.6]
	\begin{pgfonlayer}{nodelayer}
		\node [style=Z dot] (0) at (-2.25, 1) {};
		\node [style=Z dot] (1) at (-2.25, -0.75) {};
		\node [style=Z dot] (2) at (-0.75, -2.25) {};
		\node [style=Z dot] (3) at (1, -2.25) {};
		\node [style=Z dot] (4) at (2.25, -0.75) {};
		\node [style=Z dot] (5) at (2.25, 1) {};
		\node [style=none] (6) at (-3.25, -0.75) {};
		\node [style=none] (7) at (-3.25, 1) {};
		\node [style=none] (8) at (3.25, -0.75) {};
		\node [style=none] (9) at (3.25, 1) {};
		\node [style=Z phase dot] (10) at (-2.25, 2) {$\gamma_1$};
		\node [style=Z phase dot] (11) at (-2.25, -1.75) {$\gamma_2$};
		\node [style=Z phase dot] (12) at (-0.75, -3.25) {$\alpha$};
		\node [style=Z phase dot] (13) at (1, -3.25) {$\beta$};
		\node [style=none] (14) at (-2.5, 0.5) {$i_1$};
		\node [style=none] (15) at (-2.5, -0.25) {$i_2$};
		\node [style=none] (16) at (-0.25, -2.75) {$a$};
		\node [style=none] (17) at (1.5, -2.75) {$b$};
		\node [style=none] (18) at (2.5, 0.5) {$o_1$};
		\node [style=none] (19) at (2.5, -1.25) {$o_2$};
	\end{pgfonlayer}
	\begin{pgfonlayer}{edgelayer}
		\draw [style=hadamard edge] (1) to (2);
		\draw [style=hadamard edge] (2) to (3);
		\draw [style=hadamard edge] (3) to (4);
		\draw [style=hadamard edge] (1) to (4);
		\draw [style=hadamard edge] (1) to (3);
		\draw [style=hadamard edge] (3) to (0);
		\draw [style=hadamard edge] (0) to (5);
		\draw [style=hadamard edge] (0) to (4);
		\draw [style=hadamard edge] (5) to (1);
		\draw (7.center) to (0);
		\draw (6.center) to (1);
		\draw (5) to (9.center);
		\draw (4) to (8.center);
		\draw (10) to (0);
		\draw (11) to (1);
		\draw (12) to (2);
		\draw (3) to (13);
	\end{pgfonlayer}
\end{tikzpicture}
	\end{center}
	This has a focused gflow $(g, \prec)$ with $g(i_1) = \{a, o_2\}$, $g(i_2) = \{a,b, o_2\}$, $g(a) = \{b, o_2\}$ and $g(b) = \{o_1, o_2\}$ and partial order $i_1, i_2 \prec a \prec b \prec o_1, o_2$.
	Then neighbour unfusion along one of the edges $(i_2,a)$, $(a,b)$ or $(b,o_2)$ preserves gflow by Proposition~\ref{prop:neighbour-unfusion-gflow-sufficient}.
	
	On the other hand, the pair $(i_2,o_2)$ satisfies the condition $o_2\in g(i_2)$ but satisfies neither $w \prec o_2 \implies w \prec i_2$ nor $i_2 \prec w\implies o_2 \prec w$ since $a$ and $b$ sit in between the pair in the partial order.
	Applying neighbour unfusion to $i_2$ and $o_2$ does not preserve the existence of gflow since the odd neighbourhood of $\{o_1,o_2\}$ would become $\{i_2,i_2',b\}$ and thus no vertex can be corrected solely using the outputs in the resulting diagram.
	An analogous argument holds for the pair $(i_2, o_1)$ for which $o_1\notin g(i_2)$.
\end{example}

We now show that, for MBQC-form diagrams with equal numbers of inputs and outputs, the condition $b\in g(a)$ (or instead $a\in g(b)$) is necessary for neighbour unfusion to preserve gflow.

\begin{proposition}\label{prop:neighbour-unfusion-necessary}
Let $\Gamma = (G, I, O, \lambda)$ be a labelled open graph with $\abs{I}=\abs{O}$ with focused gflow $(g,\prec)$.
Suppose $a,b$ are two adjacent vertices in $G$ with $\ld(a) = \ld(b) = \XYm$.
Let $\Gamma'$ be the labelled open graph after applying neighbour unfusion to $a$ and $b$ in $\Gamma$, and suppose $\Gamma'$ has gflow.
Then we must have $b\in g(a)$ or $a\in g(b)$ for the focused gflow on $\Gamma$.
\end{proposition}

\begin{proof}
Suppose we have a pattern $\Gamma'$ with the subdiagram \eqref{eq:after_neighbour_unfusion} and assume that $\Gamma'$ has a focused gflow $(g',\prec')$.
Let $\Gamma''=(G'',I,O,\ld'')$ be the induced sub-pattern containing only those vertices of $G'$ that are either outputs or measured in the XY-plane; this must include all inputs since those cannot be measured in planes XZ or YZ.
This new measurement pattern still contains the subdiagram \eqref{eq:after_neighbour_unfusion} and it has gflow \cite[Lemma~3.15]{Backens_2021}.
In fact, since $(g',\prec')$ is focused, it implicitly follows from \cite[Proposition~3.14 and Lemma~3.15]{Backens_2021} that the gflow of the new pattern is just the restriction of the old gflow function to a smaller domain, and this is still focused; denote it by $(g'',\prec'')$.

Now every focused gflow in a pattern with only XY-plane measurements and equal numbers of inputs and outputs can be reversed in a very strict sense:
Let $\Gamma'''=(G'',O,I,\ld''')$ be the reversed pattern with the roles of inputs and outputs swapped and $\ld'''$ mapping all non-outputs to XY.
Then there exists a focused gflow $(g''_{rev}, \prec''_{rev})$ for $\Gamma'''$ where $\prec''_{rev}$ is the reverse of $\prec''$ and $u \in g''_{rev}(v)$ if and only if $v \in g''(u)$ \cite{mhalla2011graph}, cf.\ also \cite[Corollary~2.47]{Backens_2021}.

As $x$ is $\XYm$-measured and has two neighbours, to satisfy $x \in\odds{G''}{g''(x)}$ and $x \in\odds{G''}{g''_{rev}(x)}$ we require the following to hold, where $\oplus$ is the exclusive-or operator:
 \begin{center}
 	$(a \in g''(x) \wedge x \in g''(x')) \oplus (x' \in g''(x) \wedge x \in g''(a))$.
 \end{center}
 As $a'$ is also $\XYm$-measured and has two neighbours, by the same reasoning we obtain:
 \begin{center}
 	$(b \in  g''(a') \wedge x' \in g''(x)) \oplus (x\in g''(a') \wedge a' \in g''(b))$.
 \end{center} 
Then, as we cannot have both $x \in g''(a')$ and $a' \in g''(x)$, we have either that $a \in g''(x)$, $x\in g''(a')$ and $a'\in g''(b)$ or that $b \in g''(a')$, $a' \in g''(x)$ and $x \in g''(a)$ for $(g'', \prec'')$ to be a gflow for $\Gamma''$.
But $(g'', \prec'')$ is the restriction of $(g',\prec')$ to the XY-measured vertices in $\Gamma'$.
Thus either $a \in g'(x)$, $x\in g'(a')$ and $a'\in g'(b)$ or that $b \in g'(a')$, $a' \in g'(x)$ and $x \in g'(a)$.

Now, consider the following sequence of rewrites corresponding to the inverse of neighbour unfusion:
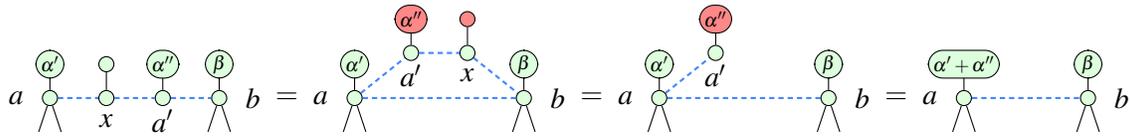
\begin{center}
	\begin{tikzpicture}[scale=0.6]
	\begin{pgfonlayer}{nodelayer}
		\node [style=Z dot] (21) at (-5.75, 0) {};
		\node [style=Z dot] (22) at (-2, 0) {};
		\node [style=none] (25) at (-6, -0.75) {};
		\node [style=none] (26) at (-5.5, -0.75) {};
		\node [style=none] (27) at (-2.25, -0.75) {};
		\node [style=none] (28) at (-1.75, -0.75) {};
		\node [style=Z dot] (29) at (-4.5, 0) {};
		\node [style=Z dot] (30) at (-3.25, 0) {};
		\node [style=Z dot] (34) at (-4.5, 0.75) {};
		\node [style=Z phase dot] (35) at (-3.25, 0.75) {$\alpha''$};
		\node [style=Z phase dot] (48) at (-5.75, 0.75) {$\alpha'$};
		\node [style=Z phase dot] (49) at (-2, 0.75) {$\beta$};
		\node [style=none] (50) at (-6.5, 0) {$a$};
		\node [style=none] (51) at (-1.25, 0) {$b$};
		\node [style=none] (52) at (-4.5, -0.5) {$x$};
		\node [style=none] (53) at (-3.25, -0.5) {$a'$};
		\node [style=none] (54) at (-0.5, 0) {$=$};
		\node [style=Z dot] (55) at (1, 0) {};
		\node [style=Z dot] (56) at (4.75, 0) {};
		\node [style=none] (57) at (0.75, -0.75) {};
		\node [style=none] (58) at (1.25, -0.75) {};
		\node [style=none] (59) at (4.5, -0.75) {};
		\node [style=none] (60) at (5, -0.75) {};
		\node [style=Z dot] (61) at (3.5, 1) {};
		\node [style=Z dot] (62) at (2.25, 1) {};
		\node [style=Z phase dot] (65) at (1, 0.75) {$\alpha'$};
		\node [style=Z phase dot] (66) at (4.75, 0.75) {$\beta$};
		\node [style=none] (67) at (0.25, 0) {$a$};
		\node [style=none] (68) at (5.5, 0) {$b$};
		\node [style=none] (69) at (3.5, 0.5) {$x$};
		\node [style=none] (70) at (2.25, 0.5) {$a'$};
		\node [style=X phase dot] (71) at (2.25, 1.75) {$\alpha''$};
		\node [style=X dot] (72) at (3.5, 1.75) {};
		\node [style=none] (73) at (6.25, 0) {$=$};
		\node [style=Z dot] (74) at (7.75, 0) {};
		\node [style=Z dot] (75) at (11.5, 0) {};
		\node [style=none] (76) at (7.5, -0.75) {};
		\node [style=none] (77) at (8, -0.75) {};
		\node [style=none] (78) at (11.25, -0.75) {};
		\node [style=none] (79) at (11.75, -0.75) {};
		\node [style=Z dot] (81) at (9, 1) {};
		\node [style=Z phase dot] (82) at (7.75, 0.75) {$\alpha'$};
		\node [style=Z phase dot] (83) at (11.5, 0.75) {$\beta$};
		\node [style=none] (84) at (7, 0) {$a$};
		\node [style=none] (85) at (12.25, 0) {$b$};
		\node [style=none] (87) at (9, 0.5) {$a'$};
		\node [style=X phase dot] (88) at (9, 1.75) {$\alpha''$};
		\node [style=Z dot] (89) at (14.5, 0) {};
		\node [style=Z dot] (90) at (17.25, 0) {};
		\node [style=none] (91) at (14.25, -0.75) {};
		\node [style=none] (92) at (14.75, -0.75) {};
		\node [style=none] (93) at (17, -0.75) {};
		\node [style=none] (94) at (17.5, -0.75) {};
		\node [style=Z phase dot] (96) at (14.5, 0.75) {$\alpha' + \alpha''$};
		\node [style=Z phase dot] (97) at (17.25, 0.75) {$\beta$};
		\node [style=none] (98) at (13.75, 0) {$a$};
		\node [style=none] (99) at (18, 0) {$b$};
		\node [style=none] (102) at (13, 0) {$=$};
	\end{pgfonlayer}
	\begin{pgfonlayer}{edgelayer}
		\draw (21) to (25.center);
		\draw (21) to (26.center);
		\draw (27.center) to (22);
		\draw (22) to (28.center);
		\draw [style=hadamard edge] (21) to (29);
		\draw [style=hadamard edge] (29) to (30);
		\draw [style=hadamard edge] (30) to (22);
		\draw (34) to (29);
		\draw (35) to (30);
		\draw (48) to (21);
		\draw (49) to (22);
		\draw (55) to (57.center);
		\draw (55) to (58.center);
		\draw (59.center) to (56);
		\draw (56) to (60.center);
		\draw (65) to (55);
		\draw (66) to (56);
		\draw [style=hadamard edge] (55) to (56);
		\draw [style=hadamard edge] (55) to (62);
		\draw [style=hadamard edge] (62) to (61);
		\draw [style=hadamard edge] (61) to (56);
		\draw (71) to (62);
		\draw (72) to (61);
		\draw (74) to (76.center);
		\draw (74) to (77.center);
		\draw (78.center) to (75);
		\draw (75) to (79.center);
		\draw (82) to (74);
		\draw (83) to (75);
		\draw [style=hadamard edge] (74) to (75);
		\draw [style=hadamard edge] (74) to (81);
		\draw (88) to (81);
		\draw (89) to (91.center);
		\draw (89) to (92.center);
		\draw (93.center) to (90);
		\draw (90) to (94.center);
		\draw (96) to (89);
		\draw (97) to (90);
		\draw [style=hadamard edge] (89) to (90);
	\end{pgfonlayer}
\end{tikzpicture}
\end{center}
where we first pivot along the edge $(x, x')$, then apply $Z$-deletion to $x'$ and finally apply the phase gadget identity rule of \cite{kissinger_reducing_2020} to add the phase of $x'$ to that of $a$. Each of these rules preserves the existence of gflow, thus the inverse of neighbour unfusion preserves the existence of gflow. Moreover, if $x \in g'(a)$, $x'\in g'(x)$ and $b'\in g'(x')$, then after applying the inverse of neighbour unfusion we get a focused gflow $(g, \prec)$ for $\Gamma$ with $b \in g(a)$ (and similarly if $a\in g'(x)$, $x\in g'(x')$ and $x' \in g'(b)$ we get a focused gflow with $a \in g(b)$).
Therefore, if the measurement pattern after neighbour unfusion has gflow, then the original pattern has a focused gflow where $b$ is in the correction set of $a$, or a focused gflow where $a$ is in the correction set of~$b$.
%
\end{proof}

An analogous argument to the above works if $b$ is an output, in which case the only option is for $b$ to be in the correction set of $a$.
Therefore the above proposition covers all the cases relevant to Staudacher et al.'s work on patterns where all measurements are in the XY-plane.

\begin{example}
    The following two measurement patterns are related by neighbour unfusion along the edge between vertices $a$ and $b$:
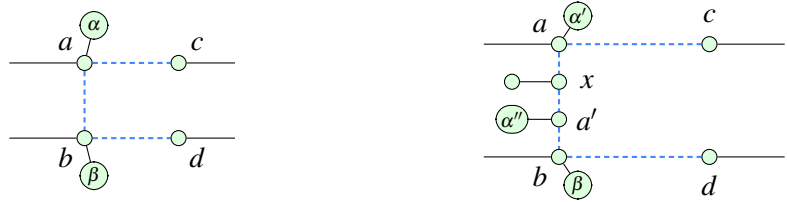
	\begin{center}
		\begin{tikzpicture}[scale=0.5, baseline=-0.1cm]
	\begin{pgfonlayer}{nodelayer}
		\node [style=none] (0) at (-4, 1) {};
		\node [style=none] (1) at (-4, -1) {};
		\node [style=Z dot] (2) at (-2, 1) {};
		\node [style=Z dot] (3) at (-2, -1) {};
		\node [style=Z dot] (4) at (0.5, 1) {};
		\node [style=Z dot] (5) at (0.5, -1) {};
		\node [style=none] (6) at (2, 1) {};
		\node [style=none] (7) at (2, -1) {};
		\node [style=Z phase dot] (8) at (-1.75, 2) {$\alpha$};
		\node [style=Z phase dot] (9) at (-1.75, -2) {$\beta$};
		\node [style=none] (10) at (-2.5, 1.5) {$a$};
		\node [style=none] (11) at (-2.5, -1.5) {$b$};
		\node [style=none] (12) at (1, 1.5) {$c$};
		\node [style=none] (13) at (1, -1.5) {$d$};
	\end{pgfonlayer}
	\begin{pgfonlayer}{edgelayer}
		\draw [style=hadamard edge] (2) to (3);
		\draw [style=hadamard edge] (3) to (5);
		\draw [style=hadamard edge] (2) to (4);
		\draw (0.center) to (2);
		\draw (4) to (6.center);
		\draw (1.center) to (3);
		\draw (5) to (7.center);
		\draw (8) to (2);
		\draw (3) to (9);
	\end{pgfonlayer}
\end{tikzpicture} \qquad\qquad\qquad\qquad
		\begin{tikzpicture}[scale=0.5, baseline=-0.1cm]
	\begin{pgfonlayer}{nodelayer}
		\node [style=none] (0) at (-4, 1.5) {};
		\node [style=none] (1) at (-4, -1.5) {};
		\node [style=Z dot] (2) at (-2, 1.5) {};
		\node [style=Z dot] (3) at (-2, -1.5) {};
		\node [style=Z dot] (4) at (2, 1.5) {};
		\node [style=Z dot] (5) at (2, -1.5) {};
		\node [style=none] (6) at (4, 1.5) {};
		\node [style=none] (7) at (4, -1.5) {};
		\node [style=Z phase dot] (8) at (-1.5, 2.25) {$\alpha'$};
		\node [style=Z phase dot] (9) at (-1.5, -2.25) {$\beta$};
		\node [style=none] (10) at (-2.5, 2) {$a$};
		\node [style=none] (11) at (-2.5, -2) {$b$};
		\node [style=none] (12) at (2, 2.25) {$c$};
		\node [style=none] (13) at (2, -2.25) {$d$};
		\node [style=Z dot] (14) at (-2, 0.5) {};
		\node [style=Z dot] (15) at (-2, -0.5) {};
		\node [style=Z phase dot] (16) at (-3.25, -0.5) {$\alpha''$};
		\node [style=Z dot] (17) at (-3.25, 0.5) {};
		\node [style=none] (18) at (-1.25, 0.5) {$x$};
		\node [style=none] (19) at (-1.25, -0.5) {$a'$};
	\end{pgfonlayer}
	\begin{pgfonlayer}{edgelayer}
		\draw [style=hadamard edge] (3) to (5);
		\draw [style=hadamard edge] (2) to (4);
		\draw (0.center) to (2);
		\draw (4) to (6.center);
		\draw (1.center) to (3);
		\draw (5) to (7.center);
		\draw (8) to (2);
		\draw (3) to (9);
		\draw [style=hadamard edge] (2) to (14);
		\draw [style=hadamard edge] (14) to (15);
		\draw [style=hadamard edge] (15) to (3);
		\draw (17) to (14);
		\draw (16) to (15);
	\end{pgfonlayer}
\end{tikzpicture}
	\end{center}
	In the first pattern, $a$ and $b$ are both inputs and thus cannot appear in correction sets.
	Hence the pattern does not have a gflow where $a$ is in the correction set of $b$ or where $b$ is in the correction set of $a$.
	Yet it does have a gflow $(g, \prec)$ with $g(a) = \{c\}$, $g(b) = \{d\}$ and $a,b \prec c,d$.

For the second pattern to have a flow $(p, \prec')$, we require $a'\in p(x)$ and $x\in p(a')$ since both vertices need to be in the odd neighbourhood of their correction set and inputs cannot appear in correction sets. This diagram can therefore not have a gflow, as the gflow conditions would require that $x \prec' a'$ and $a' \prec' x$ simultaneously, so $\prec'$ would not be strict. This diagram does have a Pauli flow however, as the $X$-measured vertex $x$ does not need to come after $a'$ in the partial order in the case of Pauli flow. The Pauli flow satisfies $p(a)=\{c\}$, $p(b)=\{d\}$, $p(x) = \{d, a'\}$ and $p(a') = \{c, x\}$ with $x\prec a,b,a'\prec c,d$.
\end{example}


The sufficient condition for neighbour unfusion to preserve gflow in Proposition~\ref{prop:neighbour-unfusion-gflow-sufficient} and the necessary condition in Proposition~\ref{prop:neighbour-unfusion-necessary} do not quite match up: we leave the question of a condition that is both necessary and sufficient to future work.

\section{Conclusion}
We have introduced several rewrite rules which preserve the existence of Pauli flow, including the first flow-preserving rewrite rule which allows us to change phases arbitrarily, rather than just by multiples of $\frac{\pi}{2}$. An immediate corollary of this rule preserving Pauli flow is that the neighbour unfusion rule of \cite{korbinian_2022} also preserves Pauli flow, potentially leading to a reduced runtime for their two-qubit gate count reduction algorithm.

At present, the circuit extraction algorithm for diagrams with Pauli flow introduces more two-qubit gates than the corresponding circuit extraction algorithm for diagrams with gflow -- future work could involve using known work on Pauli gadget optimization, such as that of \cite{Cowtan_2020}, to reduce the number of two-qubit gates obtained when performing circuit extraction on diagrams with Pauli flow. 

Other future work could involve finding an analogous result to the stabiliser completeness proof of~\cite{McElvanney_2022} for a more general fragment of the MBQC-form ZX-calculus, using Proposition~\ref{splittingproof} to introduce phases that are not just integer multiples of $\frac{\pi}{2}$. 

\section*{Acknowledgements}

We would like to thank Korbinian Staudacher and Shuxiang Cao for bringing the topic of rewriting measurement patterns to add new qubits to our attention.  Additional thanks to Korbinian Staudacher for providing a counterexample to our original claim in Section~\ref{s:neighbour-unfusion} and for countless discussions during the development of this paper. We would also like to thank Will Simmons for useful conversations on related topics.
Finally, thanks to Piotr Mitosek for suggesting the formulation of the sufficient condition for neighbour unfusion to preserve gflow, and to Simon Perdrix for pointing out an edge case in Proposition~\ref{splittingproof}, which is now treated correctly.

\bibliographystyle{eptcs}
\bibliography{refs}

\appendix

\section{Splitting a vertex preserves the existence of Pauli flow}
\label{appendix}
The following case distinction forms part of the proof of Proposition~\ref{splittingproof}. 

Let $G'$, $p'$ and $\prec'$ be defined as in Proposition~\ref{splittingproof}. We shall show that $(p', \prec')$ satisfy the nine conditions of Pauli flow.

\begin{enumerate}[label={\textit{Claim \arabic*:}}, leftmargin=*, widest=a]
	\item For all $u\in O^c$, if $v\in p'(u)$ and $u\neq v$ and $\ld'(v)\notin\{X,Y\}$, then $u\prec' v$.
	\begin{itemize}
		\item For original vertices $u\in V \setminus O$, $v \in p'(u)$ implies $v \in p(u)$ or $v \in \{x, a'\}$. If $v \in p(u)$ then either $u \prec v$ and thus $u \prec' v$ by the definition of $\prec'$ or $\ld(u)=\ld'(u)\in\{X,Y\}$. If $v = x$ then $\lambda'(v) = X$ and thus we don't need to consider this case. Finally, if $v = a'$ then $a \in p(u)$, thus $u \prec a$ which gives us $u \prec' v=a'$ by the definition of $\prec'$.
		
		\item For $u = x$, the only element of $p(x)$ is $a'$ and we have $x\prec' a'$.
		
		\item For $u = a'$,  $v \in p'(a')$ implies $v \in p(a)$ or $v = x$.  For the latter case, $\lambda'(x) = X$ and thus we do not need to consider this. In the former case, $v \in p(a)$ gives us that $a \prec v$ or $\ld(v)=\ld'(v)\in\{X,Y\}$ as $(p, \prec)$ is a Pauli flow. So either $a'  \prec v$ by the definition of $\prec'$ or the property is trivial anyway.
	\end{itemize}
	
	\item For all $u\in O^c$, if $v\in\odds{G'}{p'(u)}$ and $u\neq v$ and $\ld'(v)\notin\{Y,Z\}$, then $u\prec' v$.
	\begin{itemize}
		\item For original vertices $u\in V\setminus O$, we have defined $p'$ in such a way that $\odds{G'}{p'(u)} = \odds{G}{p(u)}$, thus this property is inherited from $(p, \prec)$ being a Pauli flow.
		
		\item For $u = x$, we have $\odds{G'}{p'(x)} = W\cup\{x\}$, and $x\prec' w$ for any $w\in W$ by the definition of $\prec'$.
		
		\item For $u = a'$, $v \in \odds{G'}{p'(a')}$ and $v \not = a'$ implies $v \in \odds{G}{p(a)}$ by \eqref{eq:2}. As $a'$ has the same successors as $a$, we get that $a' \prec' v$ as desired.
	\end{itemize}
	
	\item For all $u\in O^c$, if $\neg(u\prec' v)$ and $u \not = v$ and $\ld'(v)=Y$, then $v\in p'(u) \longleftrightarrow v\in\odds{G'}{p'(u)}$.
	\begin{itemize}
		\item For original vertices $u \in V\setminus O$, this is inherited from $(p,\prec)$, as the only changes to correction sets and odd neighbourhoods involve adding or removing $x$ or $a'$, which are not $Y$-measured.
		
		\item For $u = x$, we have $p'(x)\cup\odds{G'}{p'(x)} = W\cup\{x, a'\}$. By the definition of the partial order, $x\prec a'$ and $x\prec' w$ for all $w\in W$, so the claim is trivially true.
		
		\item For $u = a'$,
		\[
		\lambda'(v) = Y \text{ and } v\in p'(a') \text{ implies } v \in p(a).
		\]
		As $(p, \prec)$ is a focused Pauli flow, we must have $v \in \odds{G}{p(a)}$, thus $v \in \odds{G'}{p'(a)}$ and finally $v \in \odds{G'}{p'(a')}$ by \eqref{eq:2}.
		
		For the other direction,
		\[
		\lambda'(v) = Y \text{ and } v\in \odds{G'}{p'(a')} \text{ implies } v \in \odds{G}{p(a)}.
		\]
		As $(p, \prec)$ is a Pauli flow,  $v \in p(a)$ thus $v \in p'(a)$ and finally $v \in p'(a')$, as desired.
		
	\end{itemize}
	
	\item For all $u\in O^c$, if $\ld'(u)=\XYm$, then $u\notin p'(u)$ and $u\in\odds{G'}{p(u)}$.
	\begin{itemize}
		\item If $u=a'$, then this is true by inspection.
		\item If $u=x$, then the claim is true trivially since $\ld'(x)\neq\XYm$.
		\item If $u\in V\setminus O$, then this property is inherited from $(p,\prec)$.
	\end{itemize}
	
	\item For all $u\in O^c$, if $\ld'(u)=\XZm$, then $u\in p'(u)$ and $u\in\odds{G'}{p'(u)}$.
	\begin{itemize}
		\item This is trivially true as we have no $XZ$-measured vertices.
	\end{itemize}
	
	\item For all $u\in O^c$, if $\ld'(u)=\YZm$, then $u\in p'(u)$ and $u\notin\odds{G'}{p'(u)}$.
	\begin{itemize}
		\item This is trivially true as we have no $YZ$-measured vertices.
	\end{itemize}
	
	\item For all $u\in O^c$, if $\ld'(u)=X$, then $u\in\odds{G'}{p'(u)}$.
	\begin{itemize}
		\item if $u=a'$, then this is true trivially since $\ld'(a')\neq X$.
		\item If $u= x$, then the claim is true by \eqref{eq:1}.
		\item If $u\in V\setminus O$, then this property is inherited from $(p,\prec)$.
	\end{itemize}
	
	\item For all $u\in O^c$, if $\ld'(u)=Z$, then $u\in p'(u)$.
	\begin{itemize}
		\item This is trivially true as we have no $Z$-measured vertices.
	\end{itemize}
	
	\item For all $u\in O^c$, if $\ld'(u)=Y$, then $u \in p'(u)$ and $u\notin \odds{G}{p'(u)}$, or $u \notin p'(u)$ and $u\in \odds{G}{p'(u)}$.
	\begin{itemize}
		\item This is true trivially for $a$, $x$ and $a'$, and inherited for all other vertices.
	\end{itemize}
\end{enumerate}
All nine properties are satisfied, therefore $(p',\prec')$ is a Pauli flow for $G'$.

\end{document}